\documentclass[11pt]{llncs}
\usepackage{xspace}
\usepackage{graphicx, epsfig}
\usepackage{amsmath,amssymb}
\usepackage{fullpage}

\def\idrm#1{\ensuremath{\mathrm{#1}}}
\def\idtt#1{\ensuremath{\mathtt{#1}}}

\newcommand{\cM}{{\cal M}}
\newcommand{\cL}{{\cal L}}
\newcommand{\cR}{{\cal R}}
\newcommand{\cV}{{\cal V}}

\newcommand{\insp}{\idrm{insert}^+}
\newcommand{\delp}{\idrm{delete}^+}
\newcommand{\insm}{\idrm{insert}^-}
\newcommand{\delm}{\idrm{delete}^-}

\newcommand{\send}{\idtt{end}}
\newcommand{\sstart}{\idtt{start}}
\newcommand{\sleft}{\idtt{start}}
\newcommand{\sright}{\idtt{end}}
\newcommand{\sfirst}{\idtt{first}}
\newcommand{\slast}{\idtt{last}}
\newcommand{\su}{\idtt{u}}
\newcommand{\segl}{\idtt{l}}

\newcommand{\inter}{\idtt{int}}
\newcommand{\ptr}{\idrm{ptr}}

\newtheorem{invariant}{Invariant}

\newcommand{\halfrightsect}[2]{[#1,#2)}
\newcommand{\halfleftsect}[2]{(#1,#2]}

\renewenvironment{proof}{\trivlist\item[]\emph{Proof}:}%
{\unskip\nobreak\hskip 1em plus 1fil\nobreak$\Box$
\parfillskip=0pt%
\endtrivlist}

\begin{document}
%tentative title
\title{Range Reporting for Moving Points on a Grid}
\author{Marek Karpinski\inst{1},
    J. Ian Munro\inst{2},
    and Yakov Nekrich\inst{1}}
\authorrunning{M. Karpinski, J.~I. Munro and Y. Nekrich}
\institute{
    Department of Computer Science\\
    University of Bonn\\
    \email{\{marek,yasha\}@cs.uni-bonn.de}\\\mbox{}
  \and 
    Cheriton School of Computer Science\\
    University of Waterloo\\
    \email{imunro@uwaterloo.ca}
}
\maketitle
\begin{abstract}
In this paper we describe  a new data structure that supports orthogonal range 
reporting queries on a set of points that move along linear trajectories 
on a $U\times U$ grid. The assumption that points lie 
on a $U\times U$ grid enables us to significantly decrease the query time in 
comparison to the standard kinetic model. 
Our data structure  answers queries in $O(\sqrt{\log U/\log \log U}+k)$ time, 
where $k$
 denotes the number of points in the answer. The above improves over the 
$\Omega(\log n)$ lower bound that is valid in the infinite-precision 
kinetic model.
The methods used in this paper could  be also of independent interest.  
%and supports updates after 
%kinetic events in poly-logarithmic time. The total number of events 
%is bounded by $O(n^2)$.  
\end{abstract}
\section{Introduction}
\label{sec:intro}
Data structures for querying moving objects were extensively investigated 
in computational geometry and database communities. 
The orthogonal range reporting problem, i.e. the problem of storing 
a set of points $S$ in a data structure so that all points in a 
query rectangle $Q$ can be reported, was also extensively studied 
for the case of moving points.
In this paper we describe a data structure that supports range reporting 
queries for a set of moving points on a $U\times U$ grid, i.e., 
when all point coordinates are positive integers bounded by a parameter $U$.

{\bf Previous and Related Results.} 
The kinetic data structure framework  proposed by 
Basch \emph{et al.}~\cite{BGH99}  is the standard model for 
studying moving objects in computational geometry. The main idea of their 
approach is to update the data structure for a set $S$ of continuously moving 
objects only at certain  moments of time: updates are performed 
only  when certain \emph{events} changing the relevant combinatorial structure 
of the set $S$ occur. For instance, the data structure may be updated when 
the order of projections of points on the $x$-axis changes or 
the closest pair of points in $S$ changes; see e.g.,~\cite{BGH99,G98} for a 
more detailed description.

The kinetic variant of the range tree data structure
 was presented by Basch, Guibas, and Zhang~\cite{BGZ99}; their data structure 
uses $O(n\log^{d-1} n)$ space,  answers $d$-dimensional queries in 
$O(\log^{d} n+k)$ time, and can be updated after each event in 
$O(\log^{d} n)$ 
time; henceforth $k$ denotes the number of points in the answer. 
The two-dimensional data structure of Agarwal, Arge, and 
Erickson~\cite{AAE03} supports  range reporting queries 
in $O(\log n+k)$ time and uses $O(n\log n/\log \log n)$ space; the 
cost of updating their data structure after each event is $O(\log^2 n)$.
As follows from standard information-theoretic arguments, the 
$O(\log n)$ query time is optimal in the infinite-precision 
kinetic model. 
Linear space kinetic data structures were considered by Agarwal, Gao, and
 Guibas~\cite{AGG02} and Abam, de Berg, and Speckmann~\cite{AdBS09}. 
However these data structures have significantly higher query times: 
the fastest linear space construction~\cite{AdBS09} answers $d$-dimensional 
queries in $O(n^{1-1/d}+k)$ time. 

A number of geometric problems can be solved more efficiently when 
points lie on a grid, i.e.,  when coordinates of points are 
integers\footnote{For simplicity, we assume that 
all points have positive coordinates.}  bounded by a parameter $U$.
In the case of range reporting, significant speed-up can be achieved 
if the set of points $S$ does not change.
%In the static scenario, if the set of points $S$ does not change, 
There
are static  data structures that support orthogonal range reporting queries 
in $O(\log \log U +k)$ time~\cite{O88,ABR00}. 
On the other hand, %in the dynamic scenario 
if points can be inserted into or deleted from $S$, then 
any data structure that 
supports updates in $\log^{O(1)}n$ time needs $\Omega(\log n/\log \log n+k)$ 
time to answer a two-dimensional range reporting query~\cite{AHR98}. 
This bound is also valid in the case when all points belong to 
a $U\times U$ grid. 

{\bf Our Result.} 
In this paper we consider the situation  when coordinates of moving 
points belong to a $U\times U$ grid. 
Our data structure supports orthogonal range reporting queries in 
$O(\sqrt{\log U/\log \log U} + k)$ 
time. This result is valid in the standard kinetic model with additional 
conditions that all points move with fixed velocities along linear 
trajectories and all changes in the trajectories are known in advance.
Queries can be answered at any time $t$, where $t$ is a positive  integer 
bounded by $U^{O(1)}$.
Updates are performed only when  $x$- or $y$-coordinates of any two points 
in $S$ swap their relative positions, and each update takes  poly-logarithmic 
time. 
The total number of events after which the data structure must be updated 
 is $O(n^2)$. 
For instance, for  $U=n^{O(1)}$ our data structure answers queries in 
$O(\sqrt{\log n/\log \log n}+k)$ time.
Our result also demonstrates that  the lower bound for dynamic range 
reporting queries can be surpassed in the case when the set $S$ consists of 
linearly moving points. Our data structure uses $O(n\log^2 n)$ space and 
supports updates in $O(\log^3 n)$ time, but space usage and update cost
can be reduced if only special cases of reporting queries must be supported.
We describe  a $O(n)$ space data structure that supports updates 
in $O(\log n)$ time and dominance queries in 
$O(\sqrt{\log U/\log \log U} + k)$ time. We also describe  a $O(n\log n)$ 
space data structure that supports updates in $O(\log^2 n)$ time and
 three-sided\footnote{The query range of a dominance query 
is a product of two half-open intervals. The query range of a three-sided 
query is a product of a closed interval and a half-open interval.}  queries
in $O(\sqrt{\log U/\log \log U} + k)$ time.

\section{Overview}
In section~\ref{sec:onedim} we show that we can find the predecessor point of 
any $v\in U$ in the set $S$ (with respect to $x$- or $y$-coordinates) 
in $O(\sqrt{\log U/\log \log U})$ time by answering a point location query 
among a set of segments. In fact, identifying the predecessor of a 
point $q$ is the bottleneck of our query answering procedure. 

In section~\ref{sec:domin} we describe the data structure that reports 
all points $p\in S$ that dominate the query point $q$, i.e. all points 
$p$ such that $p.x\geq q.x$ and $p.y\geq q.y$; henceforth 
$p.x$ and $p.y$ denote the $x$- and  $y$-coordinates 
of a point $p$. The query time of our data structure is 
 $O(\sqrt{\log U/\log \log U} + k)$. The data structure 
is based on the modification of the $d$-approximate boundary~\cite{VV96}
for the kinetic framework. 
The $d$-approximate boundary~\cite{VV96} enables us to obtain an estimation 
for the number of points in $S$ that dominate an arbitrary point $q$. 
If a  point $q$  dominates a point on a $d$-approximate boundary $\cM$, 
then $q$  
is dominated by at most  $2d$ points of $S$; if $q$  
is dominated by a point on $\cM$, then $q$ is dominated by at least $d$ points 
of $S$. 
In section~\ref{sec:domin} we show that  a variant of a $d$-approximate 
boundary can be maintained under kinetic events.
If a query point $q$ is dominated by $k\leq \log n$ points of $S$, 
we can reduce the dominance query on $S$ to a dominance query 
on a set that contains $O(\log n)$ points using a $d$-approximate 
boundary for $d=\log n$; see section~\ref{sec:domin}. 
Otherwise, if $k> \log n$, 
 we can answer a query in $O(\log n + k)=O(k)$ time using a standard 
kinetic data structure~\cite{AAE03}.

%% Using $\cM$ we can reduce
%% It is possible to construct $O(n/d)$ sets $Dom(s)$ so that for  any point 
%% $q$ that dominates a point on $\cM$, all points from the set $S$  
%% that dominate $q$ belong to $Dom(s)$ for some $s$; see 
%% section~\ref{sec:domin}. 
%% We can determine whether a query point $q$ dominates a point on $\cM$ 
%% by answering a one-dimensional searching query in 
%% $O(\sqrt{\log U/\log \log U})$ time.  
%% If $q$ dominates some point $p$ on a $d$-approximate boundary $\cM$ for 
%% $d=\log n$, we identify the set $Dom(s)$ that 
%% consists of $O(\log n)$ points and contains all points that 
%% dominate $q$; then, we report all points that dominate $q$ using the 
%% data structure for $Dom(s)$ in $O(\log \log n +k)$ time.
%% If $q$ is dominated by a point on $\cM$, then $q$ is dominated by at least 
%% $\log n$ points of $S$. In this case  we use a dynamic data structure to 
%% report all points of $S$ that dominate $q$ in $O(\log n + k)=O(k)$ 
%% time. In both cases the query time is $O(\sqrt{\log U/\log \log U} + k)$.

A data structure that supports dominance queries can be transformed 
into a data structure that supports arbitrary orthogonal range reporting
queries~\cite{CG86b,SR95} by dividing the set $S$ into subsets $S_i$ 
and constructing dominance data structures for each $S_i$ as described 
in section~\ref{sec:orth}. However, 
we may have to delete a point $p$ from one subset $S_i$ and insert it into 
a subset $S_j$ after a kinetic event. 
Unfortunately, the construction of~\cite{VV96} is static. 
It is not clear how (and whether) to modify the $d$-approximate 
boundary, so that insertions and deletions are supported. 
However, in our case the deleted (inserted) point always 
has the maximal or minimal  $x$- or $y$-coordinate among all points in 
$S_i$ ($S_j$). 
We will describe in section~\ref{sec:orth} how our dominance data structure
 can be modified to support these special update operations without 
increasing the query time. Our technique is similar to the 
logarithmic method and can be of independent interest. 
Thus we obtain the data structure for 
general orthogonal range reporting queries.
\section{One-Dimensional Searching}
\label{sec:onedim}
Let $S_x(t)$ and $S_y(t)$ denote the sets of $x$- and $y$-coordinates of 
all points at time $t$. 
In this section we will describe how we can identify the predecessor 
of any $q_x$ in $S_x(t)$ (resp.\ of $q_y$ in $S_y(t)$) at current time 
$t$ in $O(\sqrt{\log U/\log\log U})$ time using a linear space data structure.

Let $x_i(t)=a_it+b_i$ be the equation that describes the $x$-coordinate 
of the point $p_i\in S$ at time $t$. 
The trajectory of the point in $(t,x)$ plane ($t$-axis is horizontal) 
is a sequence of segments. 
Since we assume that all changes of point trajectories are known in 
advance, endpoints of  all segments are known in advance. 
%Changes of the $x$-coordinate of a point $p_i$ with time 
%correspond to a segment on the plane $(t,x)$.
Two points swap ranks of their $x$-coordinates at time $t$ 
if and only if their segments intersect at time $t$. 
We can find intersection points of all segments using the standard sweepline 
algorithm~\cite{BO79} in $O((n+f)\log n)$ time, where $f$ is the number 
of segment intersections. We start the sweepline at $t=0$ and move 
it to the right until $n$ intersection points are identified 
or the last intersection point is found. 
These intersection points and the corresponding segments induce  a
subdivision of the $(x,t)$ plane of size $O(n)$. 
We can construct the data structure for planar point 
location~\cite{CP09} that supports queries in $O(\sqrt{\log U/\log \log U})$ 
time. 
Let $t_{max}$ be the largest $t$-coordinate of the already processed 
intersection point. For $t< t_{max}$, we can find the predecessor of
any $x$ by locating the segment lying immediately below the point $(t,x)$. 
When $t=t_{max}$, we continue the sweepline algorithm and find 
the next $n$ segment intersection points. 
The algorithm described in~\cite{BO79} finds  $f$ next segment intersection 
points $O(f\log n)$ time. Since  the point location data structure~\cite{CP09}
for a subdivision of size $f$ can be constructed in $O(f)$ time,
an amortized cost of processing a kinetic  event is $O(\log n)$.
We can de-amortize the update cost using standard techniques.
\section{Dominance Queries}
\label{sec:domin}
In this section we describe the data structure that reports 
all points from $S$ that dominate the query point $q$, 
i.e. all points in the region 
$\halfrightsect{q.x}{+\infty}\times\halfrightsect{q.y}{+\infty}$.
Our data structure is based  on maintaining the $d$-approximate 
boundary for a set $S$.
The notion of a $d$-approximate boundary is introduced in~\cite{VV96}; in 
this paper we  change the definition and describe a kinetic version of this 
construction. 

For a horizontal segment $s$, we denote by $\sstart(s)$ and  $\send(s)$ 
 $x$-coordinates of the left and the right endpoint of $s$; 
we denote by $y(s)$ the $y$-coordinate of all points on $s$.
We will say that a segment $s$ covers a point $p$ 
if the $x$-coordinate of $p$ belongs to $[\sstart(s),\send(s)]$.
In this paper we define a  $d$-approximate boundary as a polyline $\cM$ that consists 
of alternating horizontal and  vertical segments, divides the plane into 
two parts, and  satisfies the following properties:
\begin{invariant}\label{inv:segm}
Let  $s$ and $r$ be two consecutive horizontal 
segments. Then $|\{p\in S\,|\, \sstart(s) \leq p.x \leq \send(r) \}| > d/2$. 
\end{invariant}
\begin{invariant}\label{inv:lpoint}
When a new segment $s$ is inserted, the left endpoint of $s$ is dominated 
by  at most $3d/2$ points of $S$. The number of points in $S$ 
that dominate the left endpoint of a segment $s\in \cM$ does not exceed 
$2d$.  
\end{invariant}
\begin{invariant}\label{inv:rpoint}
When a new segment $s$ is inserted, the right endpoint of $s$ is dominated 
by  at least $d$  points. The number of points in $S$ 
that dominate the right endpoint of a segment $s\in \cM$ remains constant.  
\end{invariant}
Using Invariants~\ref{inv:segm}-\ref{inv:rpoint}, we can prove the following
 Lemma.
\begin{lemma}
Every point on a $d$-approximate boundary $\cM$ is dominated by 
 at least $d$ points and at most  $2d$ points of $S$. There are $O(n/d)$ 
horizontal segments in $\cM$. 
\end{lemma}
\begin{proof}
If a point $p\in \cM$ is dominated by $k$ points from $S$, then 
the left endpoint of some segment $s$ is dominated by at least $k$ 
points and the right endpoint of some segment $r$ is dominated by at most 
$k$ points. Hence, it follows from Invariants~\ref{inv:lpoint} and \ref{inv:rpoint} that $d\leq k < 2d$. 
By Invariant~\ref{inv:segm}, for two consecutive segments $r$ and $s$ 
there are more than $d/2$ points $p\in S$, such that $p.x$ belongs to the 
interval $[\sstart(r),\send(s)]$. 
For each point $p$, $p.x$ belongs to at most two such intervals; hence,
 the total number of segments is less than $8n/d$. 
\end{proof}
An example of a $d$-approximate boundary is shown on
 Fig~\ref{fig:tapprox-example}. 
\begin{figure}[bt]
  \centering
\includegraphics[width=.5\textwidth]{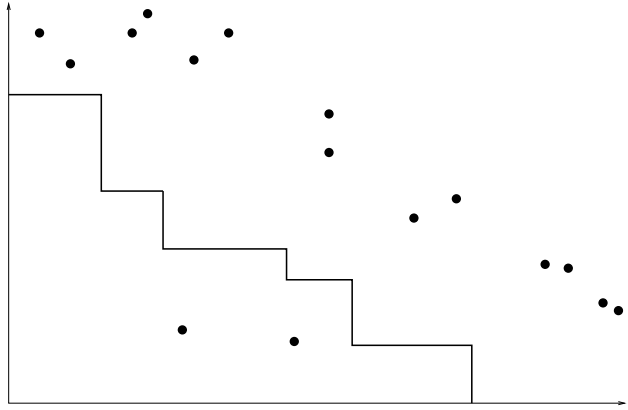}
\caption{An example of a $d$-approximate boundary for $d=6$.}
\label{fig:tapprox-example}
\end{figure}
We will show below how the concept of a $d$-approximate boundary can be used 
to support dominance queries in $O(\sqrt{\log U/\log \log U} + k)$ time. 
Later in this section we will show how Invariants~\ref{inv:segm}
-\ref{inv:rpoint} can be maintained. \\
{\bf Kinetic Boundary.}
We will use a kinetic variant of the $d$-approximate boundary, i.e. 
segments of the boundary move together with  points of $S$. 
For every horizontal segment $s$ in a boundary $\cM$, let $\segl(s)$ denote the 
point with the largest $y$-coordinate such that $\segl(s).y < y(s)$ and 
let $\su(s)$ denote the point with the smallest $y$-coordinate such that 
$\su(s).y > y(s)$. 
Let $\sfirst(s)$ denote the point with the smallest $x$-coordinate
such that $\sfirst(s).x > \sstart(s)$; let $\slast(s)$ denote the point with
the largest $x$-coordinate such that $\slast(s).x < \send(s)$. 
We assume that $\sstart(s)=\sfirst(s)-\frac{1}{2}$ and 
$y(s)=\su(s)-\frac{1}{2}$. That is, the left end and the $y$-coordinate 
of a segment change when $\su(s)$ and $\sfirst(s)$ move. The right 
end of the previous segment and the $x$-coordinate of the connecting 
vertical segment change accordingly. 

{\bf Answering Queries.} 
Our data structure is based on a $d$-approximate boundary $\cM$ of $S$ 
for $d=\log n$. 
For each segment $s\in \cM$ we maintain the set $Dom(s)$ of all points 
that dominate the left endpoint of $s$. Obviously, the set $Dom(s)$ 
changes only when events concerning $\sfirst(s)$  or 
$\su(s)$ take place. 

All points of $Dom(s)$ are stored 
in a  data structure $D_s$ that supports dominance queries 
in $O(\log d +k)=O(\log \log n+k)$ time. We can use the data structure 
of~\cite{AAE03}, so that the space usage is $O(d)$ 
and updates after events are supported in $O(\log \log n)$ time. 
It is possible to modify the data structure of~\cite{AAE03}, so that 
points can be inserted into $Dom(s)$ or deleted from $Dom(s)$ in 
$O(\log \log n)$ time. 
 All points of $S$ are also stored in a kinetic  data structure 
$G$ that uses $O(n)$ space and supports dominance queries in $O(\log n +k)$ 
time and updates after kinetic events in $O(\log n)$ time. 
Again, we can use the result of~\cite{AAE03} to implement $G$.
Finally, we must be able to identify  for each point $p\in S$
the segment $s\in \cM$ that covers $p$. Using the 
dynamic union-split-find data structure of~\cite{DR91} or the van Emde Boas 
data structure~\cite{E77}, we can find the segment 
that covers any $p\in S$ in $O(\log \log n)$ time. 
When a new segment is inserted into or deleted from $\cM$, the 
data structure is updated in $O(\log \log n)$ time. 

Given a query point $q$, we identify the point $p$ with the largest 
$x$-coordinate such that $p.x\leq q.x$; this can be done 
in $O(\sqrt{\log U/\log \log U})$ time using the construction described in 
section~\ref{sec:onedim}. 
The point  $q$ dominates a point on $\cM$ if and only if $q$ dominates the 
left endpoint of the segment $s$ that covers $p$ or the left 
endpoint of the segment $h$ that follows $s$. 
Suppose that $q$ dominates a point on $\cM$. Then 
$q$ is dominated by at most $2d$ points 
of $S$. Let $v$ be the left endpoint of a segment $g$, 
such that $v$ is dominated by $q$. 
Each point  $p\in S$ that dominates $q$  also dominates $v$; hence, 
all points that dominate $q$ belong to $Dom(g)$. 
We can use the data structure $D_g$ 
 and report all points that dominate $q$ in $O(\log \log n + k)$ time. 
Suppose that $q$ does not dominate any point on $\cM$. Then there is at least 
one point on $\cM$ that dominates $q$. Hence, there are  
$k\geq d= \log n$ points of $S$ that dominate $q$. 
Using data structure $G$, we can report all points that dominate 
$q$ in $O(\log n + k)=O(k)$ time. 
\\{\bf Maintaining the $d$-Approximate Boundary.} 
It remains to show how to maintain Invariants~\ref{inv:segm}-\ref{inv:rpoint} 
after operations $x$-move and $y$-move. 
Henceforth, we will use the following notation. 
Suppose that $p.x< q.x$ before some kinetic event and $p.x>q.x$ after 
this event. Then, we say that $p$ is $x$-moved behind $q$ ($q$ is $x$-moved 
before $p$).
Suppose that $p.y< q.y$ before some kinetic event and $p.y>q.y$ after 
this event. Then, we say that $p$ is $y$-moved above $q$ ($q$ is $y$-moved 
below $p$). Each kinetic event can be represented as a combination of at most
 one $x$-move and at most one $y$-move.

First,  we consider the Invariant~\ref{inv:lpoint}. Suppose that 
the left endpoint of an interval $s$ is dominated by $2d$ points 
of $S$. Let $h$ be the segment that precedes $s$, i.e., $\send(h)=\sstart(s)$. 
Let $v$ be the vertical segment that connects $h$ and $s$. 
We look for a  point $p$ with $p.x=\sstart(s)$ and $y(s)< p.y < y(h)$ 
such that $p$ is dominated by $3d/2$ points. If such a point 
$p$ exists, we set $p_l=p$ and search for a point $p_r$ with  
$\sstart(s) < p_r.x < \send(s)$ and $p_r.y=p_l.y$ such that $p_r$ is dominated 
by $d$ points of $S$. If there is no such point, i.e., if the point 
$(\send(s),p_l.y)$ is dominated by at least $d$ points, we 
replace $s$ with a segment $s'$ such that  the left endpoint of $s'$
is $p_l$ and the right endpoint of $s'$ is the point $(\send(s),p_l.y)$;
see Fig~\ref{fig:lpoint}a. 
In other words we change the $y$-coordinate of $s$ to $p.y$.
The new set $Dom(s)$ contains $3d/2$ points and can be constructed in 
$O(d)$ time.  
If there is a point $p_r$ with  
$\sstart(s) < p_r.x < \send(s)$ and $p_r.y=p_l.y$ that  is dominated 
by $d$ points of $S$, then  we replace $s$ with two new segments 
$s'$ and $s''$. The left and right 
endpoints of $s'$ are $p'_l$ and $p'_r$ respectively.
The left endpoint of $s''$ is the point $p''$ such that $p''.x=p_r.x$ 
and $p''.y=y(s)$. See Fig~\ref{fig:lpoint}b.
The set $Dom(s')$ contains $3d/2$ points. 
There are at most $d/2$ points $q$ such that $q.x > \sstart(s)$ and 
$y(s)< q.y \leq p.y$; hence, there are  at most $d/2$ points $q$ such that
 $q.x \geq  p_r.x$ and $y(s)< q.y \leq p.y$. Therefore, since 
$p_r$ is dominated by $d$ points of $S$, $p''$ is dominated by at most 
$3d/2$ points of $S$ and $Dom(s'')$ contains at most $3d/2$ points. 
Since $p''$ is dominated by the right endpoint of the segment  $s$, $Dom(s'')$ 
contains at least $d$ points. 
We can construct $Dom(s')$ and $Dom(s'')$ and data structures 
$D_{s'}$ and $D_{s''}$ in $O(d)$ time. 

If the right endpoint of $h$ is dominated by at least $3d/2$ points, 
we shift the vertical segment $v$ in $+x$ direction, so that 
the right endpoint of $h$ is dominated by $d$ points from $S$ 
or the segment $s$ is removed. That is, we identify the point $r$ with 
$r.y=y(h)$ such that either $r.x=\send(s)$ and $r$ is dominated by at least 
$d$ points of $S$ or $r.x< \send(s)$ and $r$ is dominated 
by  $d$ points of $S$. 
If $r.x< \send(s)$, we set $\send(h)=\sstart(s)=r.x$ and 
update $Dom(s)$, $D_s$ accordingly ($O(d)$ points are removed 
from $Dom(s)$ and $D_s$). See Fig~\ref{fig:lpoint}c.
The new left endpoint of $s$ is dominated by at most $3d/2$ points.
  If $r.x=\send(s)$, we remove the segment $s$ with $D_s$ and $Dom(s)$. 

The update procedure removes at most one segment and inserts at 
most two  new segments that satisfy the Invariant~\ref{inv:lpoint}. 
Hence, the update procedure takes $O(d)=O(\log n)$ time. 

\begin{figure}[tb]
  \centering
  \begin{tabular}{ccccc}
  \includegraphics[width=.25\textwidth]{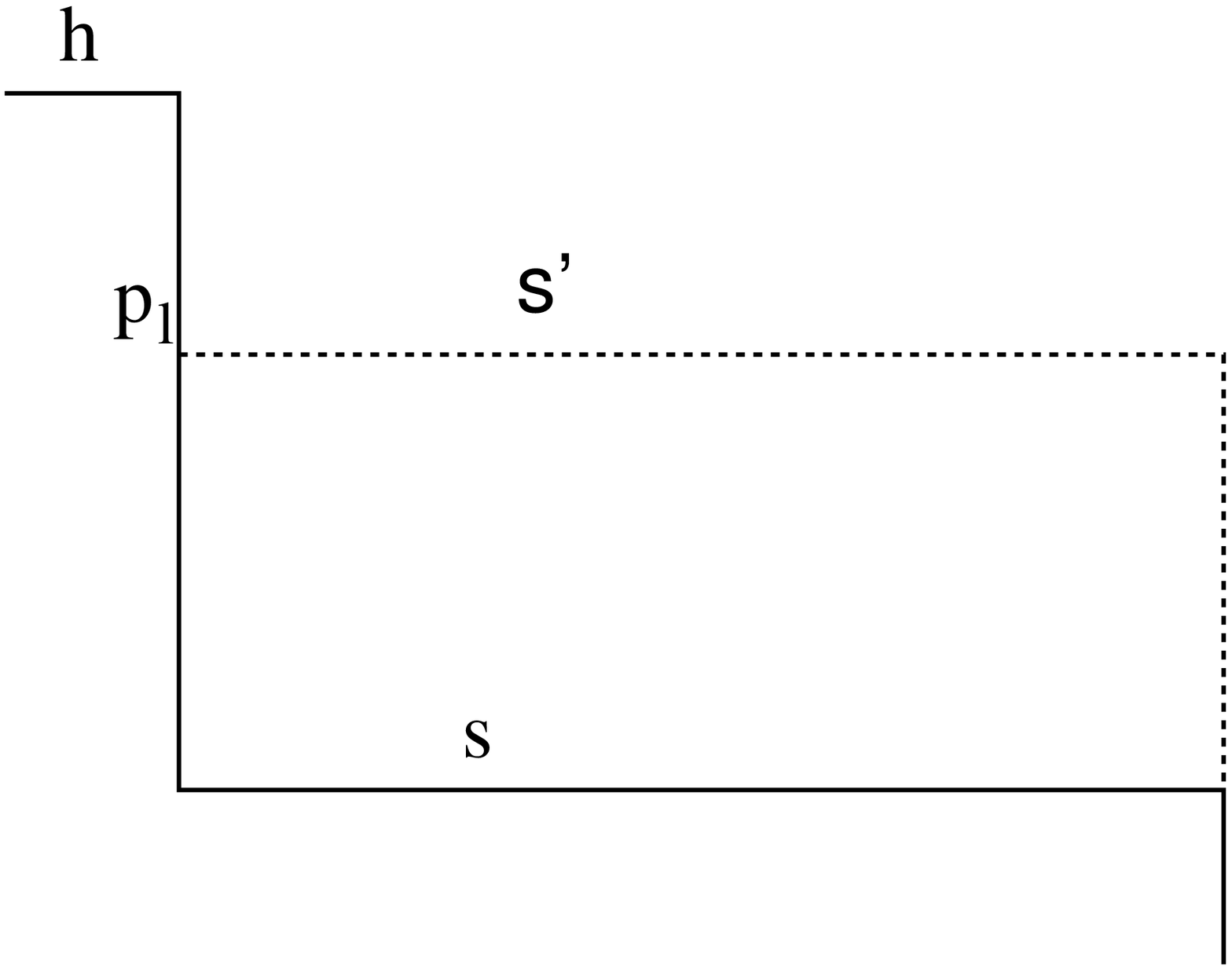} & \hspace*{.2cm} &
  \includegraphics[width=.25\textwidth]{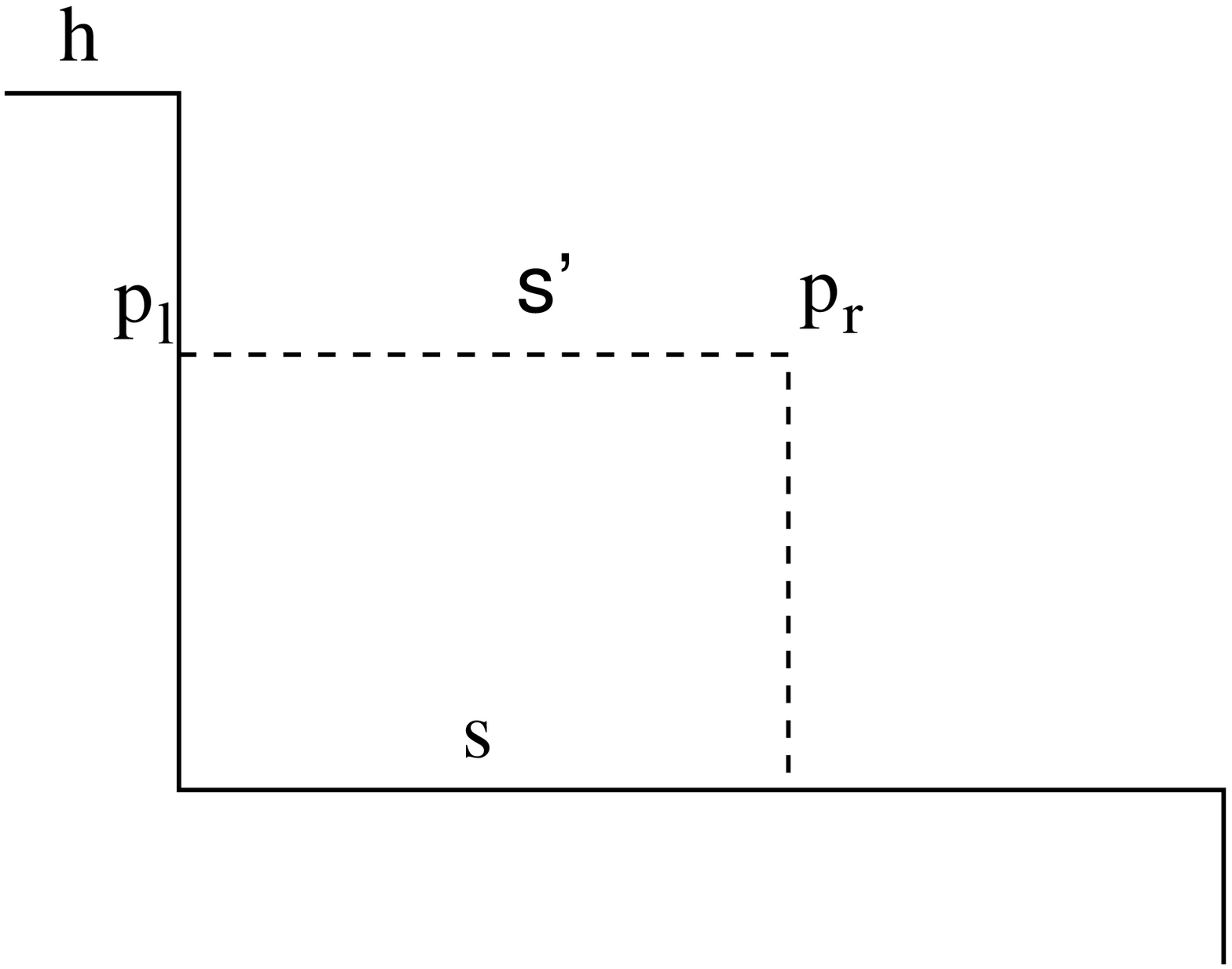} & \hspace*{.2cm} &
  \includegraphics[width=.25\textwidth]{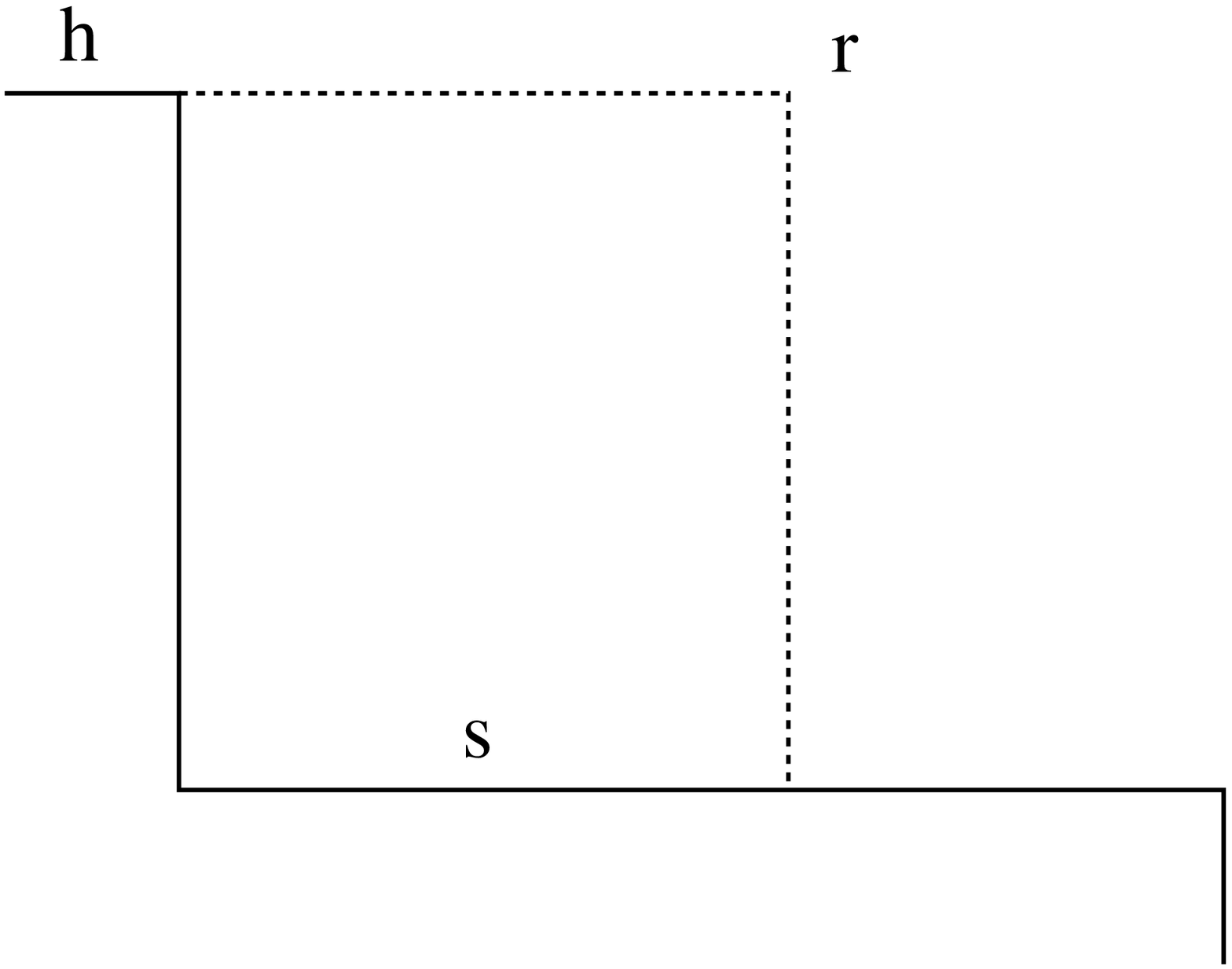}  \\
  {\bf (a)}   &  & {\bf(b)} & & {\bf (c)} \\
  \end{tabular}
  \caption{Updating the $d$-approximate boundary when the 
    Invariant~\ref{inv:lpoint} is violated. For simplicity point of $S$ are
    not shown.
    {\bf (a)} There exists a point $p_l$ but there is no $p_r$. 
    {\bf (b)} There exist both $p_l$ and $p_r$. 
    {\bf (c)} There are neither $p_l$ nor $p_r$.
     }
  \label{fig:lpoint}
\end{figure}

Now we describe how the Invariant~\ref{inv:segm} can be maintained. 
Let $h$ and $s$ be two consecutive segments ($\send(h)=\sstart(s)$)
and suppose that there are $d/2$ points that belong to 
$[\sstart(h), \send(s)]$. 
We replace $h$ and $s$ with one new segment $g$ as follows. 
If the left endpoint of $h$ is dominated by at least $3d/2$ 
points, then we set $\send(h)=\send(s)$ and remove the segment 
$s$, $Dom(s)$, and $D_s$. 
The point $q$ with $q.y=y(h)$ and $q.x=\send(s)$ is the new right 
endpoint of $h$. Since there are at most $d/2$ points $p\in S$ such that 
$\sstart(s) \leq p.x \leq q.x$, $q$ is dominated by at least 
$d$ points of $S$. Hence, the new segment $h$ satisfies 
Invariant~\ref{inv:rpoint}. See Fig.~\ref{fig:segm}a.
If the point $l$ with $l.x=\sstart(h)$ and $l.y=y(s)$ is dominated by 
at most  $3d/2$ points  of $S$, we set $\sstart(s)=\sstart(h)$ 
and remove $h$, $Dom(h)$ and $D_h$. The set $Dom(s)$ and the data 
structure $D_s$ are updated. See Fig.~\ref{fig:segm}b. 
Since there are less than $d/2$ points
$p\in S$, such that $p.y > y(s)$ and $\sstart(h) \leq  p.x \leq \send(h)$, 
$O(d)$ new points are inserted into $Dom(s)$ and $D_s$. 
If the left endpoint of $h$ is dominated by less than $3d/2$ points 
and the point $l$ is dominated by more than $3d/2$ points, then 
we replace $h$ and $s$ with a new segment $g$. The left endpoint 
of $g$ is the point $m$ such that $m.x =\sstart(h)$, 
$y(s) < m.y < y(h)$, and $m$ is dominated by $3d/2$ points. 
The right endpoint of $g$ is the point $r$ with $r.x=\send(s)$ and 
$r.y=m.y$. See Fig.~\ref{fig:segm}c. 
The point $r$ is dominated by at least $d$ points of $S$. 
Hence, $g$ satisfies Invariants~\ref{inv:lpoint} and~\ref{inv:rpoint}. 

\begin{figure}[tb]
  \centering
  \begin{tabular}{ccccc}
  \includegraphics[width=.25\textwidth]{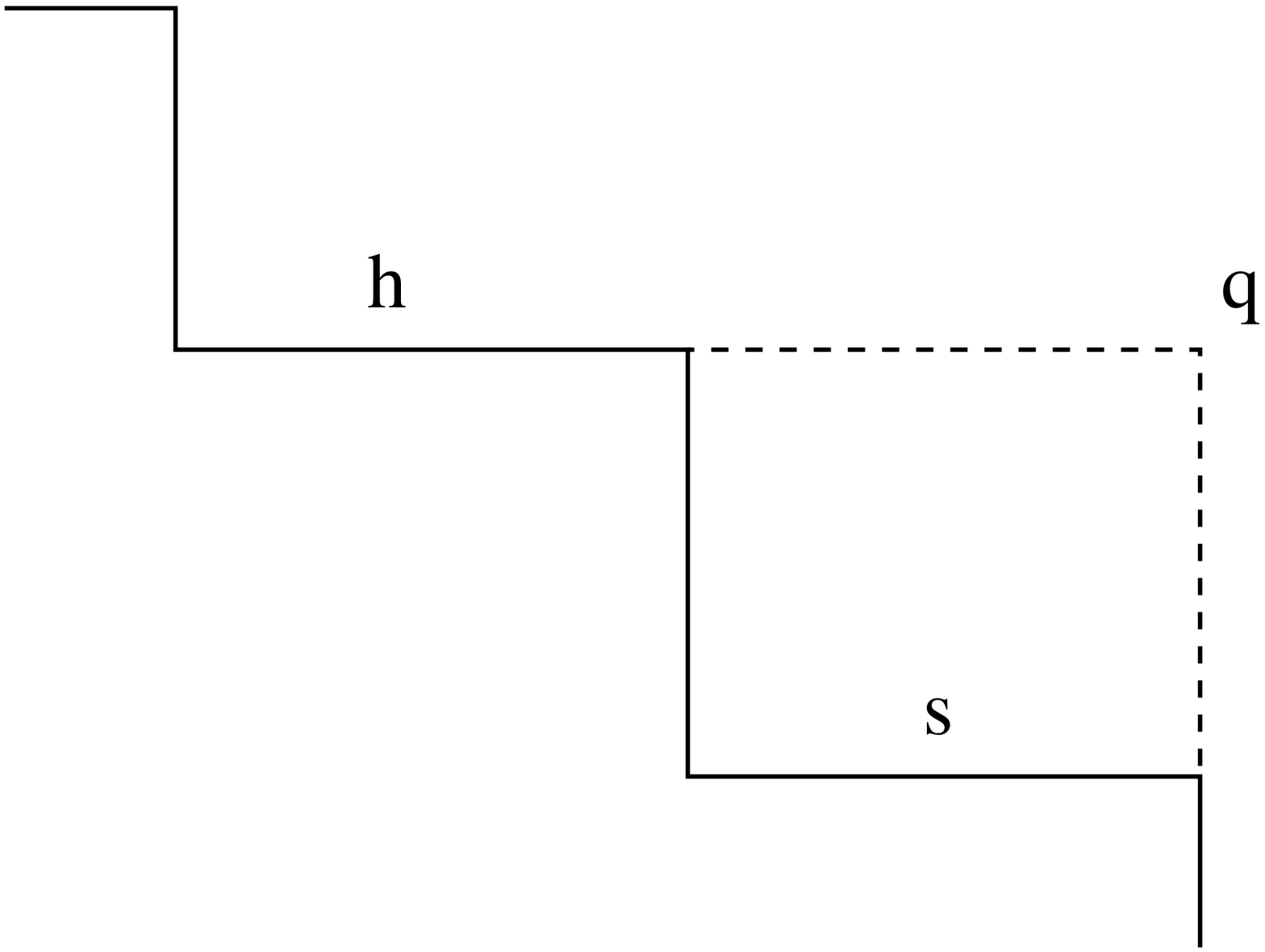} & \hspace*{.2cm} &
  \includegraphics[width=.25\textwidth]{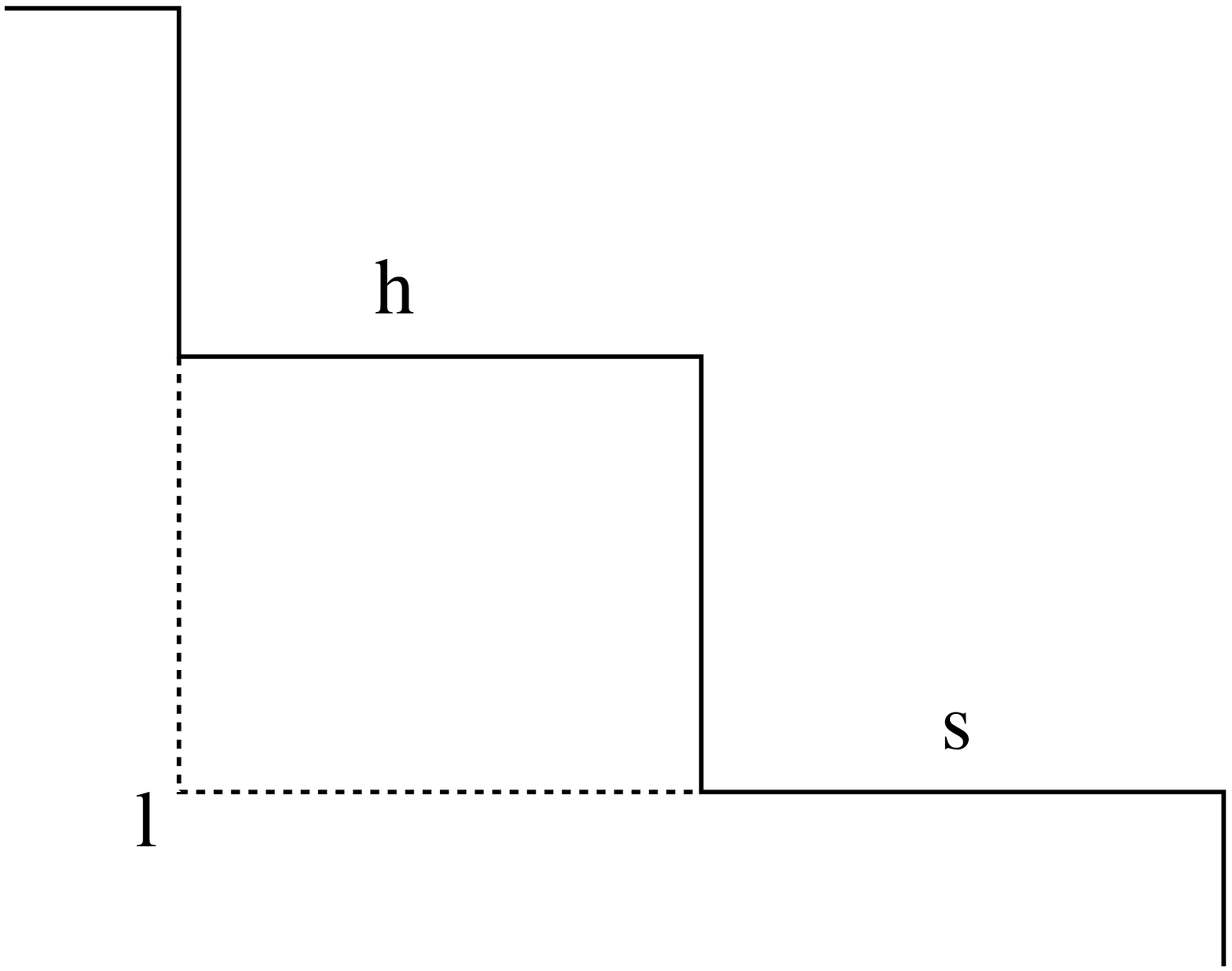} & \hspace*{.2cm} &
  \includegraphics[width=.25\textwidth]{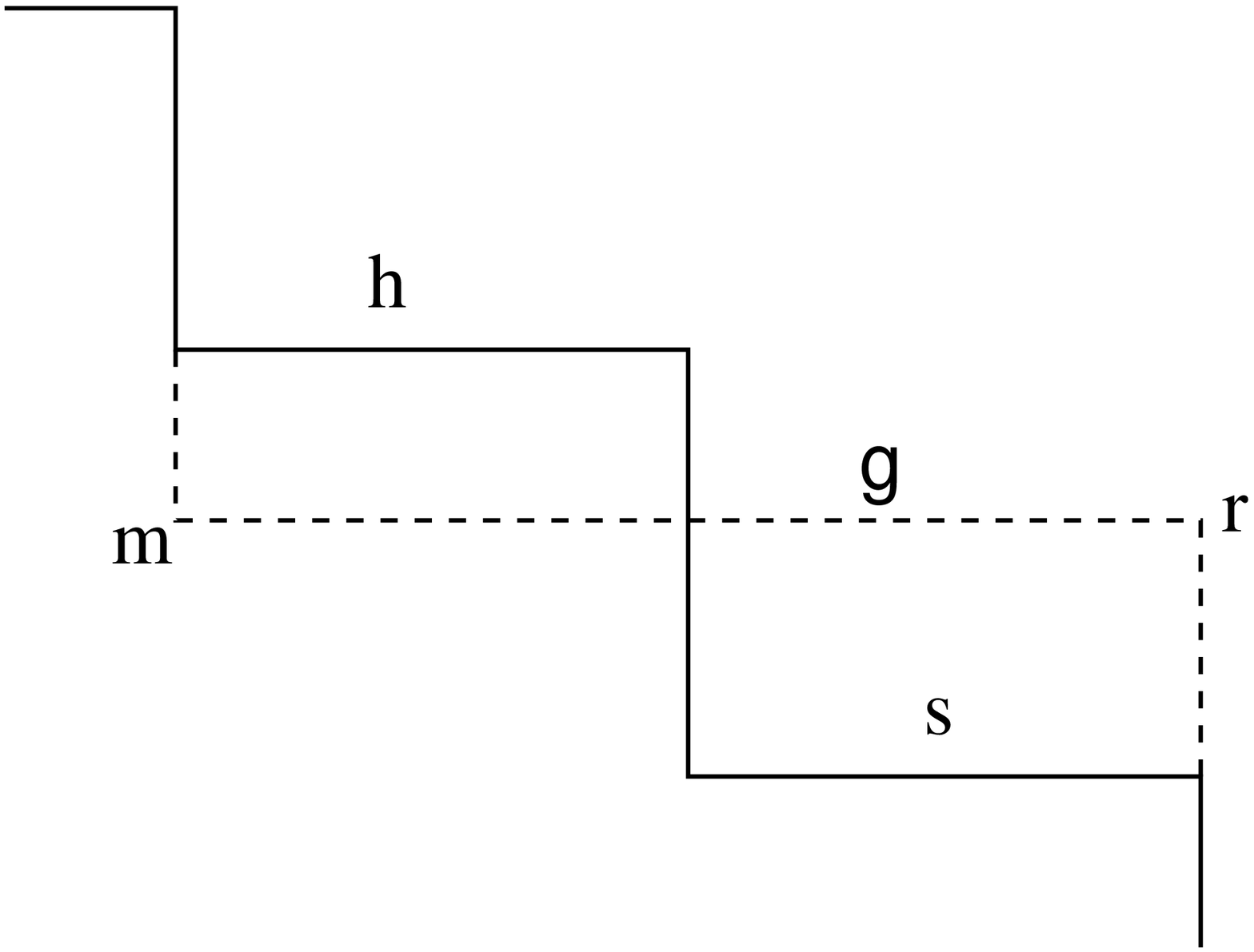}  \\
  {\bf (a)}   &  & {\bf(b)} & & {\bf (c)} \\
  \end{tabular}
  \caption{Updating the $d$-approximate boundary when the 
    Invariant~\ref{inv:segm} is violated. For simplicity points of set $S$ 
    are not shown.
    {\bf (a)} The left endpoint of $h$ is dominated by $\geq 3d/2$ points. 
    {\bf (b)} The point $l$ is dominated by at most $3d/2$ points. 
    {\bf (c)} The left endpoint of $h$ is dominated by less than
    $3d/2$ points but $l$ is dominated by more than $3d/2$ points.
     }
  \label{fig:segm}
\end{figure} 

Now we turn to Invariant~\ref{inv:rpoint}.
The number of points that dominate the right endpoint of  a horizontal segment 
$s$ changes: 1) if a point $p$ with $p.x> \send(s)$ is $y$-moved above $\su(s)$
 or below $\su(s)$; 2) if a point $p$ with $p.y > y(s)$ is $x$-moved 
before $\slast(s)$ or behind $\slast(s)$. 
%Essentially, after a kinetic event 
%the segment $s$ or its left endpoint is shifted so that the number of points 
%in $D_s$ remains unchanged or increases by $1$.

First, we consider $y$-moves. Essentially, when a point $p$ is $y$-moved, 
we shift the segment $s$ in $+y$ or $-y$
 direction so that 
the number of points that dominate the right endpoint of $s$ remains 
unchanged. Suppose that 
 a point $p=\su(s)$ with $p.x > \sright(s)$ is $y$-moved below $\segl(s)$;
let  $y_n$ denote the  $y$-coordinate of $p$.\\
(a) If $\sleft(s) < \segl(s).x < \sright(s)$, then we change the $y$-coordinate 
of   $s$ so that  $y(s)=y_n-\frac{1}{2}$. The old point $\segl(s)$ is added to 
$Dom(s)$ and $D_s$. \\
(b) If $\segl(s).x < \sleft(s)$, then we also change the $y$-coordinate of   
$s$ so that \\ $y(s)=y_n-\frac{1}{2}$. \\
(c) If $\segl(s).x > \sright(s)$, then we  change the $y$-coordinate of   
$s$ so that $y(s)=\segl(s).y-\frac{1}{2}$. We delete $p$ from $Dom(s)$ and $D_s$ and 
insert $\segl(s)$ into $Dom(s)$ and $D_s$.\\
% We don't need this (d) and (f) in the kinetic model - Yakov
%(d) If a point $p'$ is moved below $\su(s)$, then we  change the $y$-coordinate
% of   $s$ so that $y(s)=y_n-\frac{1}{2}$. \\
Suppose that  a point $p$ with $p.x > \sright(s)$ is $y$-moved above  $\su(s)$\\
(d) If $\sleft(s) < \su(s).x < \sright(s)$, then we change the $y$-coordinate
 of   
$s$ so that $y(s)=q.y -\frac{1}{2}$, where $q$ is the point with the smallest 
$y$-coordinate such that $q.y>y_n$. 
The old point $\su(s)$ is removed from 
$Dom(s)$ and $D_s$. If $q=\su(h)$ for the segment $h$ that precedes $s$, 
then  the segment $s$ and $Dom(s)$ are deleted.
% i.e., all points of $Dom(s)$ are added to 
%$Dom(h)$. Using the union-split-find data structure, this can be done 
%in $O(\log \log n)$ time.  
\\
(e) If $\su(s).x < \sleft(s)$, then we also change the $y$-coordinate of   
$s$ so that $y(s)=q.y-\frac{1}{2}$, where $q$ is the point with the 
smallest $y$-coordinate such that $q.y>y_n$. If $q=\su(h)$ for the segment 
$h$ that precedes $s$, then  the segment $s$ and $Dom(s)$ are deleted.
 \\
(f) If $\su(s).x > \sright(s)$, then we  change the $y$-coordinate of   
$s$ so that $y(s)=y_n-\frac{1}{2}$. We delete $\su(s)$ from $Dom(s)$ 
and $D_s$ and insert $p$ into $Dom(s)$ and $D_s$.\\
% We don't need (d) and (h) in the kinetic model - Yakov
%(h) If a point $p$ is $y$-moved above  $\segl(s)$ and $y_n$ is larger than the 
%$y$-coordinate of $s$, then we change the $y$-coordinate of   
%$s$ so that $y(s)=y_n-\frac{1}{2}$.\\
Observe that the number of points in $D_s$ remains unchanged or increases 
by $1$. See Fig~\ref{fig:ymove}. 
\begin{figure}[tb]
  \begin{tabular}{ccccc}
  \includegraphics[width=.25\textwidth]{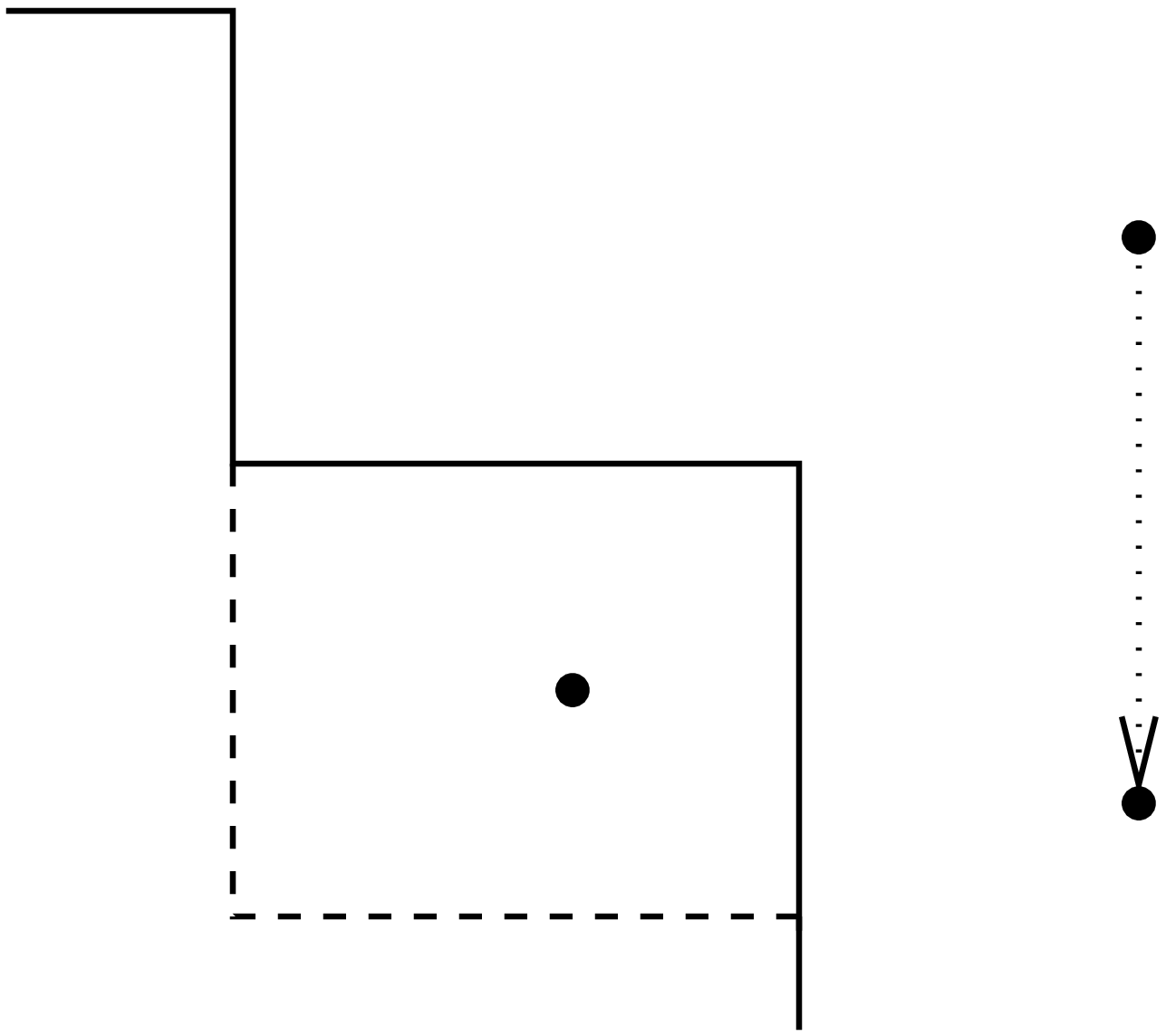} & \hspace*{.9cm} &
  \includegraphics[width=.25\textwidth]{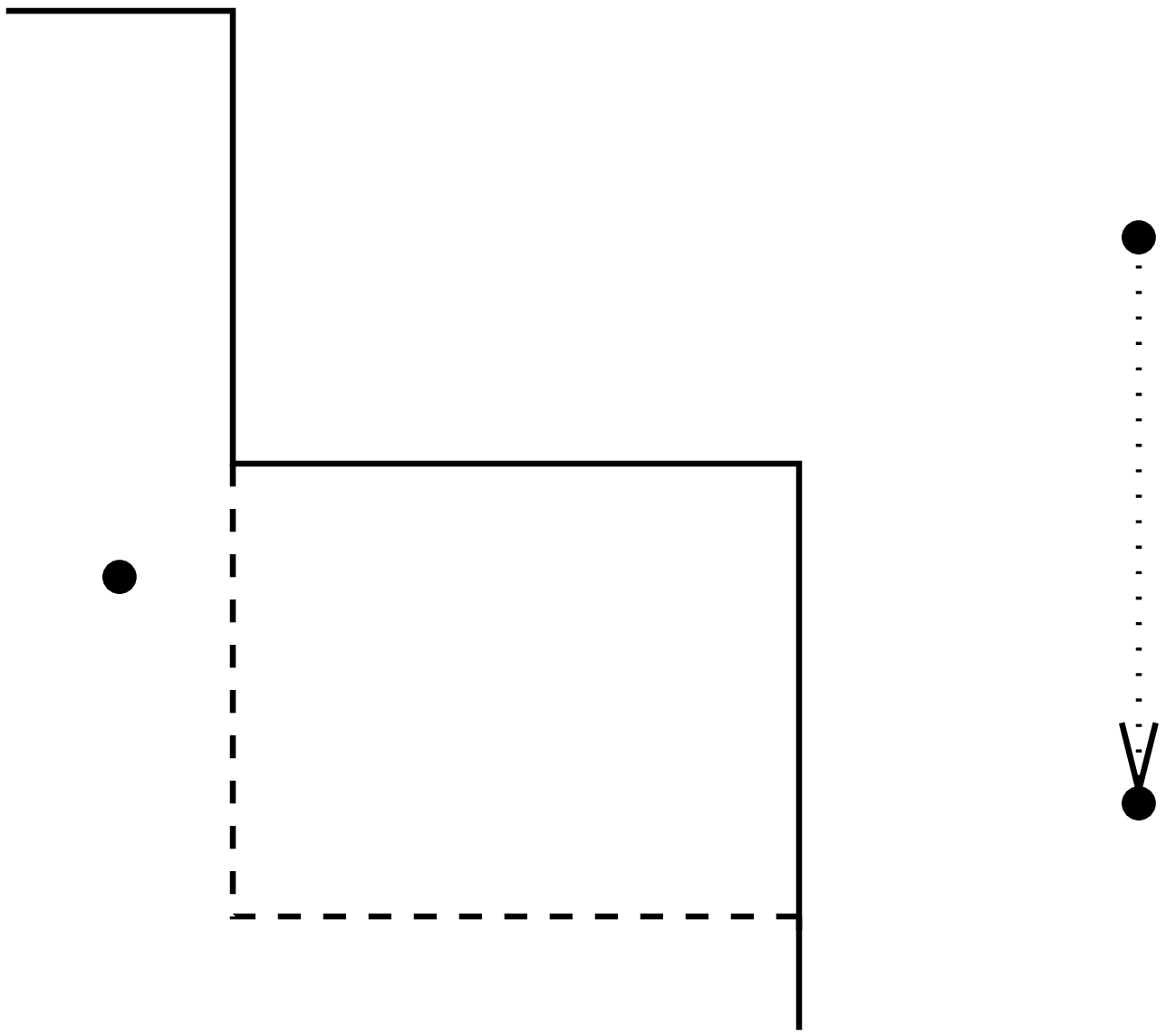} & \hspace*{.9cm} &
  \includegraphics[width=.25\textwidth]{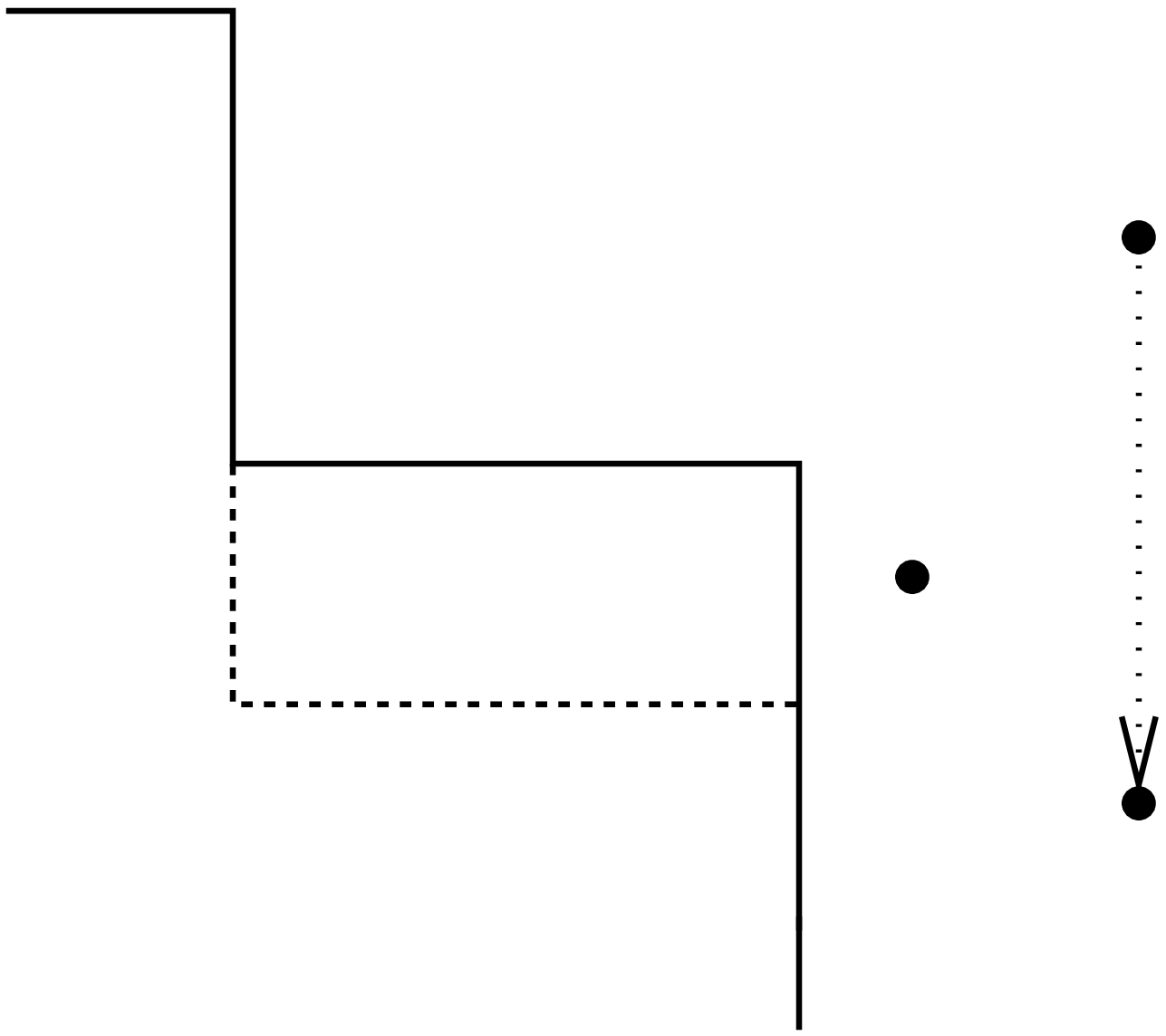}  \\
  {\bf (a)}   &  & {\bf(b)} & & {\bf (c)} \\
  \includegraphics[width=.25\textwidth]{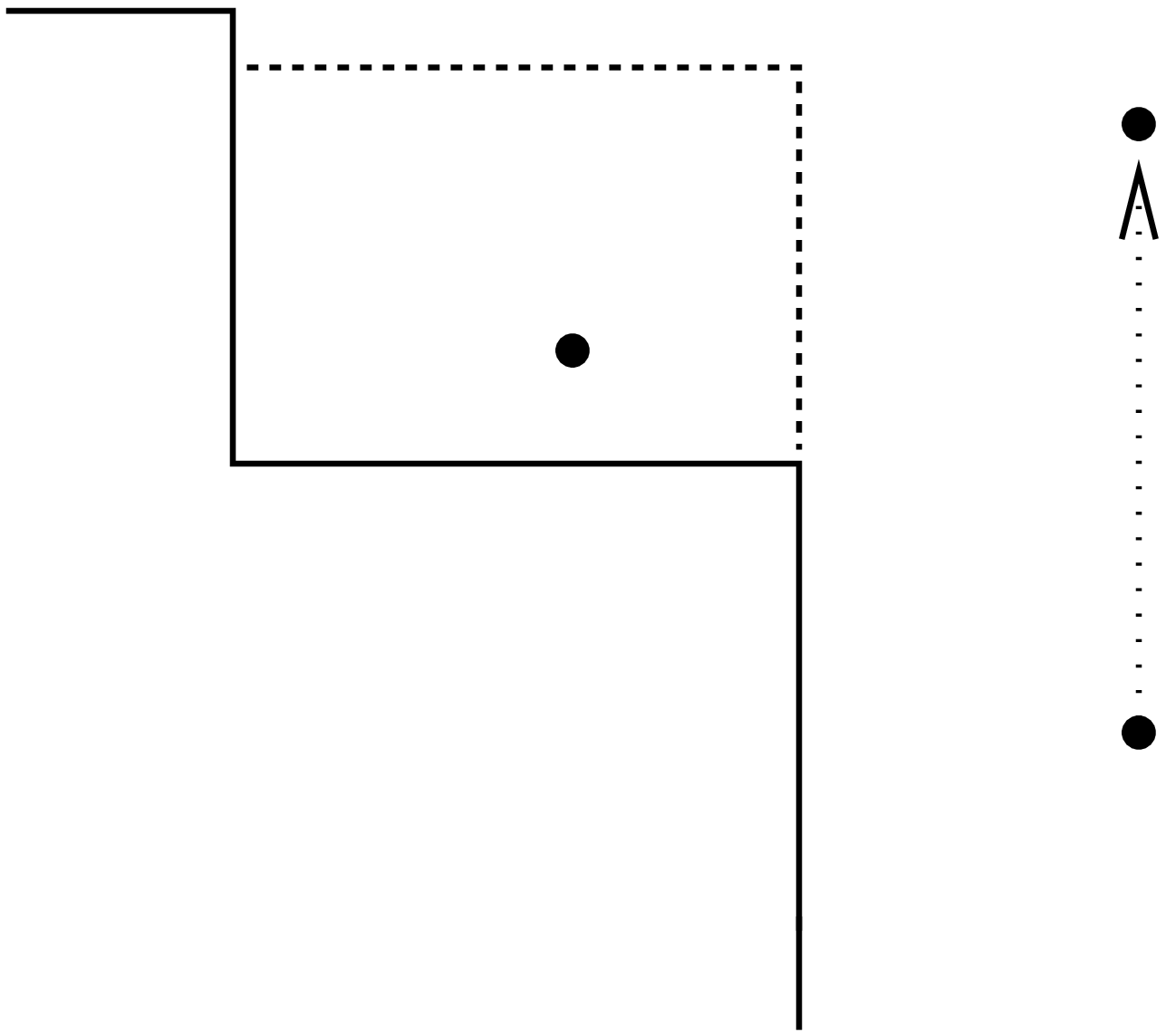} & \hspace*{.3cm} &
  \includegraphics[width=.25\textwidth]{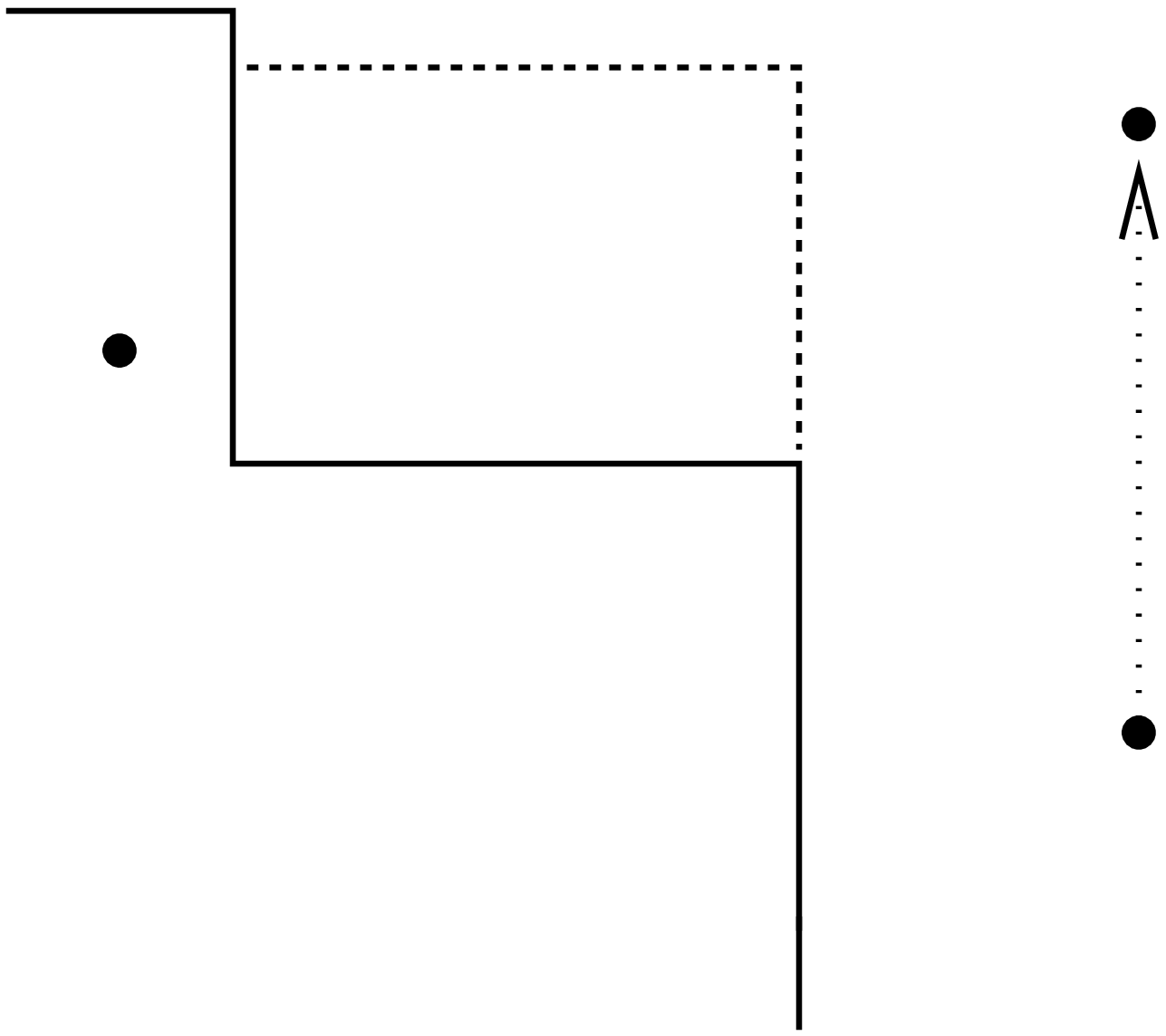} & \hspace*{.3cm} &
  \includegraphics[width=.25\textwidth]{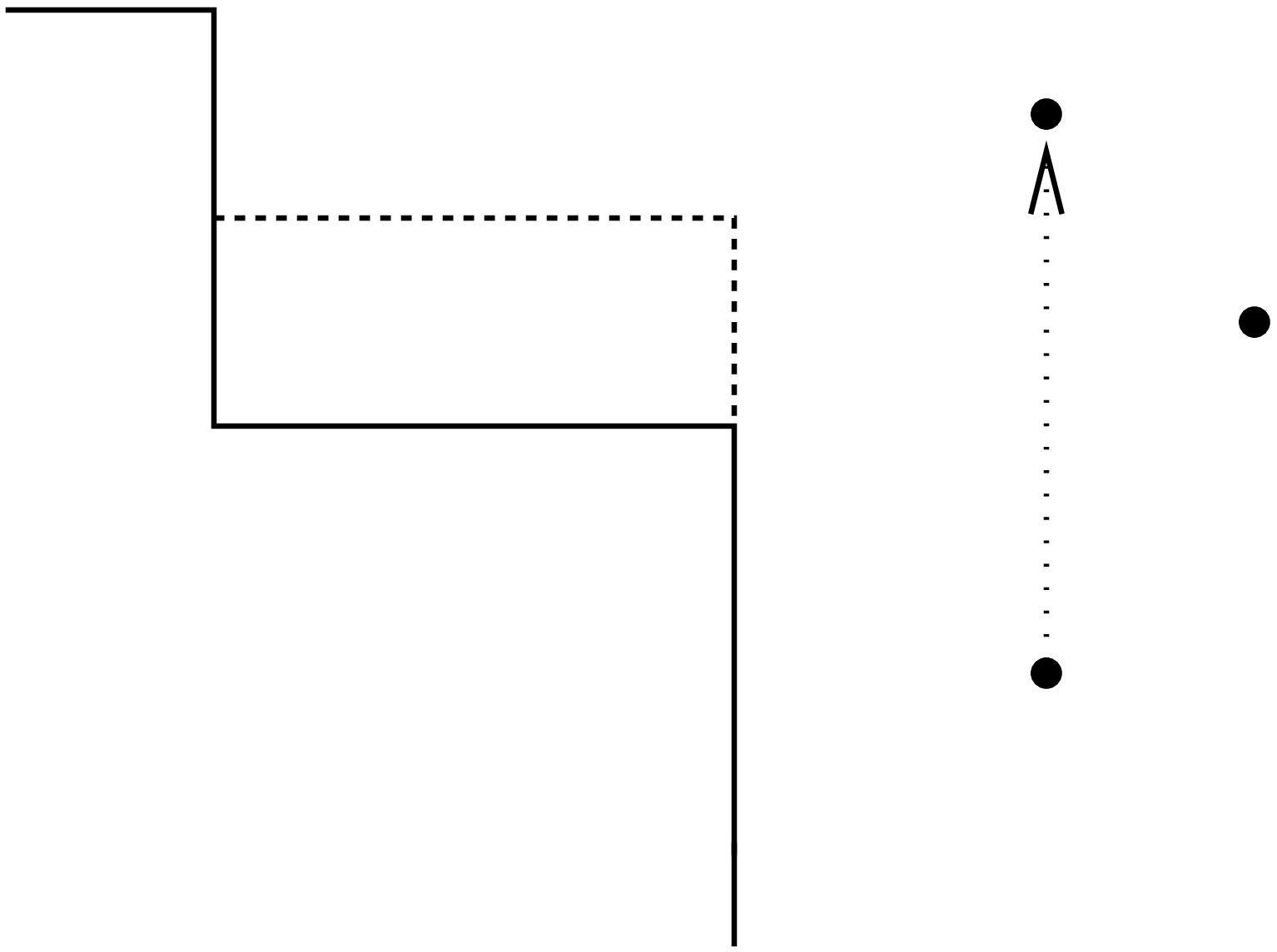}  \\
  {\bf (d)}   &  & {\bf(e)} & & {\bf (f)} \\
  \end{tabular}
  \caption{The segment $s$ is shifted so that Invariant~\ref{inv:segm} 
    is maintained after a $y$-move. Figures {\bf (a)},  {\bf (b)},{\bf (c)}, {\bf (d)}, {\bf (e)},  and {\bf (f)} correspond to cases (a), (b), (c), (d), (e), and (f) respectively.
     }
  \label{fig:ymove}
\end{figure}

We can handle the $x$-moves in a similar way. Let $h$ be the horizontal 
segment that follows $s$ in $\cM$, i.e., $\send(s)=\sstart(h)$. 
When a point $p$ is $x$-moved 
we change $\send(s)$ (and $\sstart(s)$) so that the number of points that 
dominate the right endpoint of $s$ remains unchanged. 
Suppose that 
 a point $p=\slast(s)$ with $p.y > y(s)$ is $x$-moved behind  $\sfirst(h)$;
let $x_n$ denote the new $x$-coordinates of $p$ and let $q$ be the point 
with the smallest $x$-coordinate such that $q.x>x_n$.\\
(a) If $y(h) < \sfirst(h).y < y(s)$, then we set $\sstart(h)=\send(s)=q.x-\frac{1}{2}$ and remove $\sfirst(h)$ from $Dom(h)$ and $D_h$. If $q=\sfirst(h')$ 
for the horizontal segment $h'$ that follows $h$, then we delete 
the segment $h$  and $Dom(h)$. \\
(b) If $\sfirst(h).y >  y(s)$, then we set $\sstart(h)=\send(s)=x_n-\frac{1}{2}$; we also  remove $\sfirst(h)$ from $Dom(h)$ and $D_h$ and insert 
$p$ into $Dom(h)$ and $D_h$.\\
(c) If $\sfirst(h).y <  y(h)$, then we set 
$\sstart(h)=\send(s)=q.x- \frac{1}{2}$. If $q=\sfirst(h')$ 
for the horizontal segment $h'$ that follows $h$, then we delete 
the segment $h$  and $Dom(h)$.\\
% we don't need this (d) and (h) in the kinetic model - Yakov
%(d) If a point $p'$  with $p'.y > y(s)$ is $x$-moved behind  
%$\slast(s)$ and  $x'_n> \send(s)$, 
%then we  set $\sstart(h)=\send(s)=q'.x-\frac{1}{2}$, where $q'$ is the point 
%with the smallest $x$-coordinate such that $q'.x>x'_n$. \\
Suppose that  a point $p$ with $p.y > y(s)$ is $x$-moved before  $\slast(s)$.\\
(d) If $y(h) < \slast(s).y < y(s)$, then we set $\sstart(h)=\send(s)=x_n-\frac{1}{2}$;
we add $\slast(s)$ to  $Dom(h)$ and $D_h$. \\
(e) If $\slast(s).y >  y(s)$, then we set $\sstart(h)=\send(s)=\slast(s).x-\frac{1}{2}$;
we also  remove $p$ from $Dom(h)$ and $D_h$ and insert 
$\slast(s)$ into $Dom(h)$ and $D_h$.\\
(f) If $\slast(s).y <  y(h)$, then we set $\sstart(h)=\send(s)=x_n-\frac{1}{2}$.
We add $\slast(s)$ to $Dom(h)$ and $D_h$  \\
% We don't need (h) in the kinetic model - Yakov
%(h) If a point $p$  with $p.y > y(s)$ is $x$-moved before   $\sfirst(h)$ and
% $x_n < \send(s)$, 
%then we  set $\sstart(h)=\send(s)=x_n-\frac{1}{2}$. \\
Again  the number of points in $D_s$ remains unchanged or increases 
by $1$.

%We describe   how the Invariant~\ref{inv:rpoint} is maintained in Appendix A.
%Now we turn to Invariant~\ref{inv:rpoint}.
%The number of points that dominate the right endpoint of  a horizontal segment 
%$s$ changes: 1) if a point $p$ with $p.x> \send(s)$ is $y$-moved above $\su(s)$
% or below $\su(s)$; 2) if a point $p$ with $p.y > y(s)$ is $x$-moved 
%before $\slast(s)$ or behind $\slast(s)$. 
%After a kinetic event 
%the segment $s$ or its left endpoint is shifted so that the number of points 
%in $D_s$ remains unchanged or increases by $1$.
%Hence, we can maintain Invariant~\ref{inv:rpoint}  by $O(1)$ updates
% of data structures $D_s$ after each event. A detailed description can be 
% found in Appendix A.

 We will show in the full version 
that we update $\cM_i$ because Invariants~\ref{inv:lpoint} or~\ref{inv:segm} 
are violated at most 
once for $\Omega(d)$ events. Update procedures for maintaining 
Invariants~\ref{inv:segm} and~\ref{inv:lpoint} involve inserting and deleting a constant number of segments 
into $\cM$ and the data structure $D_s$ for every segment $s$ contains 
$O(d)$ points. Hence, we must perform $O(1)$ amortized  updates of data 
structures $D_s$ after each kinetic event.  Since every update of 
$D_s$ takes $O(\log \log n )$ time, the amortized cost of updating after an 
event is $O(\log \log n)$. Since each data structure $D_s$ can be constructed 
in $O(\log n)$ time, the worst-case update time is $O(\log n)$. \\
{\bf Construction of a $d$-Approximate Boundary.} 
Now we show that $d$-approximate boundary can be constructed in 
$O(n)$ time if points are sorted by $x$- and $y$-coordinates. 
Since a $d$-approximate boundary consists of alternating horizontal 
and vertical segments, it suffices to determine the endpoints of horizontal 
segments.
We can guarantee that the left endpoint of each segment 
is dominated by $3d/2$ points of $S$ 
and the right endpoint of each segment  is dominated by $d$ points of $S$ 
using the following algorithm.
Lists  $L_x$ and $L_y$ contain points of $S$ sorted in descending 
order of their $x$-coordinates and $y$-coordinates respectively. 
With every element of $L_x$ we store a pointer to its position in $L_y$
and vice versa. 
The pointer $\ptr_x$ ($\ptr_y$) points to the first not yet 
processed element in $L_x$ ($L_y$). 
We construct a sequence of horizontal and vertical segments so that 
the left endpoint of each segment is dominated by $3d/2$ points and the 
right endpoint of each segment is dominated by $d$ points. 

We assign $\ptr_y$ to the $(3d/2+1)$-th element of $L_y$ and $\ptr_x$ to the first element of $L_x$; the point $p_l$ 
with $p_l.y=\ptr_y.y-1/2$ and $p.l.x=0$ is the left endpoint of the first 
segment. Clearly, $p_l$ is dominated by $3d/2$ points. 
(1) The right endpoint of the segment with left endpoint 
$p_l$ can be found as follows. 
We traverse elements of $L_x$ 
that follow $\ptr_x$ until $d/2+1$ points $q$ such that $q.y> \ptr_y.y$ 
are visited and update $\ptr_x$ accordingly.  
%When a point $p$ stored in  $L_X$ is  visited, we remove the 
%entry with the point $p$ from $L_y$. 
The point $p_r$ with $p_r.x=\ptr_x.x-1/2$
and $p_r.y=p_l.y$ is the right endpoint of the currently constructed 
horizontal  segment.
(2) We identify the left endpoint of the next segment by traversing 
elements of $L_y$ that follow $\ptr_y$ until $d/2$ points $q$ such that 
$q.x > \ptr_x.x$ are visited or we reach the end of the list  $L_y$. 
The pointer $\ptr_y$ is updated accordingly.
If we reached the end of $L_y$, then all points are processed and the 
algorithm is completed.
Otherwise we set $p_l.x=\ptr_x.x-1/2$ and $p_l.y= \ptr_y.y-1/2$, go to 
step (1) and determine the right endpoint of the next segment. 
\begin{theorem}
There exists a linear  space data structure that supports 
dominance queries for a set of linearly moving points 
on $U\times U$ grid in $O(\sqrt{\log U/\log \log U}+k)$
time and updates after kinetic events in $O(\log n)$ worst-case time.
The amortized cost of updates is $O(\log \log n)$.  
If trajectories of the points do not change, then the total number of kinetic
 events is bounded by $O(n^2)$. 
\end{theorem}

 \section{Orthogonal Range Reporting Queries}
\label{sec:orth}
 Three-sided range reporting queries and orthogonal range reporting 
 queries can be reduced to dominance queries with help of standard techniques. 
 However to apply these techniques in our scenario, we must modify the data 
 structure  of section~\ref{sec:domin}, so that insertions and deletions 
are supported in some special cases.
At the end of this section we demonstrate how our data structure with 
additional operations can be used to support arbitrary orthogonal range 
reporting queries. 

 {\bf Additional Update Operations.}
 We first describe the data structure that supports insertions and deletions 
 in two special cases:
 Let $x_{\min}$ and $x_{\max}$  be the smallest and the largest  $x$-coordinates
 of  points in $S$.
 The operation  $\insp_x$ inserts a point $p$ with $x_{\max} <  p.x$. 
 The operation $\delp_x$ deletes a point $p$ with $p.x=x_{\max}$.
 Operations $\insm_x$ and $\delm_x$ insert and delete a point whose 
 $x$-coordinate is smaller than $x$-coordinates of all other points in $S$. 
 It is easy to augment our data structure so that $\insm_x$ and $\delm_x$
 are supported: the inserted (deleted) point is either below a $d$-approximate 
 boundary $\cM$ or 
 dominates only the leftmost horizontal segment of $\cM$. 
 Hence, each $\insm_x$ and $\delm_x$  affects the data structure $D_s$ 
 for at most one segment $s$. Essentially we can handle 
 $\insm_x$ and $\delm_x$ in the same way as $x$-moves for the leftmost 
 segment $s$.
 Maintaining the $d$-approximate boundary $\cM$ after $\insp_x$ and $\delp_x$ 
 is more involved: 
 since the $y$-coordinate of a newly inserted (deleted) point can be 
 larger than the $y$-coordinates of all (other) points in $p$, we may have 
 to update  data structures $D_s$ for all segments $s\in \cM$
 after a single update operation. Below we describe how $\insp_x$ and 
 $\delp_x$ can be supported.

Our approach is similar to the logarithmic method~\cite{B79,O87} that is used 
to transform static data structures into data structures that support 
insertions.
 We construct the  data structure of section~\ref{sec:domin} augmented 
 with $\insm_x$ and $\delm_x$  for sets $H_2,H_3,\ldots, H_m$. 
 A point $p$ $x$-overlaps with point $q$ if $p.x >  q.x$. A point $p$ 
$x$-overlaps with a set  $S$ if it $x$-overlaps with at least one
 point  $q\in S$. Each set $H_i$ satisfies the following conditions:
 \begin{enumerate}
 \item
 For $i=2,\ldots,m-1$, $H_i$ contains between $2^{2i-1}$ and $2^{2i+2}$ points;
 $H_m$ contains between $2^{2i-2}$ and $2^{2i+2}+2^{2i}$ points
 \item
 Each point of  $H_{i}$  $x$-overlaps at most $2^{2i-4}$ points in  $H_{i-1}$ 
 \item
At most $2^{2i-3}$ points from  $H_i$ $x$-overlap with $H_{i-1}$
 \end{enumerate}
 As follows from conditions 2 and 3, no element of $H_{i+1}$ $x$-overlaps with 
$H_{i-1}$: The rightmost point in $H_{i+1}$ $x$-overlaps at most 
$2^{2i-2}$ leftmost points in $H_i$ by condition 2. Only $2^{2i-3}$ 
rightmost points in $H_i$ can $x$-overlap with a point in $H_{i-1}$. 
Hence, any $q\in H_{i+1}$ that $x$-overlaps a point in $H_{i-1}$ would 
$x$-overlap $|H_i|-2^{2i-3}>2^{2i-2}$ points in $H_i$, which contradicts 
condition 2. 

For $j=1,\ldots,m$, we maintain $\max_y(j)=\max\{p.y|p\in H_{j}\}$
and $\min_x(j)=\min\{p.x|p\in H_{j}\}$. 
All $\min_x(j)$ are stored in a kinetic binary search tree. 
For each $j$, $2\leq j \leq m$, a point $p_j$ such that $p_j.y=\max_y(j)$ 
and $p_j.x=j$ is stored in a kinetic  data structure $Y$; since $Y$ 
contains $O(\log n)$ points, 
$Y$ supports dominance reporting queries in $O(\log \log n + k)$ time.
The data structure $D_j$, $j=2,\ldots, m$,
 contains all points from $H_j$ and supports
dominance queries, $x$-moves and $y$-moves as described in 
section~\ref{sec:domin}. 
To speed-up update operations, we store only one data structure 
$G$ for all points in $\cup_{j=2}^m H_j$. $G$ is implemented as described in~\cite{AAE03}, supports dominance 
queries in $O(\log n +k)$ time and can be modified to support 
arbitrary updates as well as kinetic events in $O(\log n)$ time. 
We also store  one 
data structure $\cV$ that contains $x$-coordinates of all points 
$\cup_{j=2}^m H_j$ and enables us to search in the set of $x$-coordinates 
at any time $d$. That is, all $D_j$ share one data structure $G$ and one 
data structure $\cV$.  
For each $j$ we also store all points of  $H_j$ in a list $L_j$ that contains 
all points from $H_j$ in the descending order of their $y$-coordinates.

Given a query $Q= \halfrightsect{a}{+\infty}\times\halfrightsect{b}{+\infty}$,
 we can find in $O(\log\log n)$ time the 
smallest  index $j$, such that  at least one point in $H_j$ has 
$x$-coordinate smaller than $a$. Then, as follows from conditions 2 and 3, 
 $x$-coordinates of all points in $H_2\cup\ldots\cup  H_{j-1}$ 
are greater than $a$, and both
  $H_{j}$ and $H_{j+1}$ may contain points whose 
$x$-coordinates are greater  than or equal to  $a$.
Sets $H_{f}$, $f>j+1$, 
contain only points whose $x$-coordinates are smaller than 
$a$.
Using $Y$, we can identify all $H_f$ such that $f< j$ and $H_f$ contains 
at least 
one point $p$ with $p.y \geq b$. For every such $f$ all points $p$ 
such that $p\in H_f$ and $p.y \geq b$ can be reported by traversing the list 
$L_f$. Hence, reporting all points  $p\in H_f$ such that $f< j$, 
$p.x \geq a$ and $p.y\geq b$ takes $O(\log \log n + k)$ time. 
We can report all points in $H_j$ and $H_{j+1}$ that belong to 
$Q= \halfrightsect{a}{+\infty}\times\halfrightsect{b}{+\infty}$ 
in $O(\sqrt{\log U/\log \log U} +k)$ time using data structures 
$D_j$ and $D_{j+1}$ respectively.

It remains to show how conditions 1-3 above can be maintained.
Clearly, conditions 1-3 are influenced by $x$-moves and operations 
$\insm_x$, $\delm_x$, $\insp_x$ and $\delp_x$; $y$-moves cannot violate them.  
We say that a set $S$ is \emph{$x$-split} into sets $S_1$ and $S_2$ if 
$S_1\cup S_2=S$ and the $x$-coordinates of all points in $S_1$ are larger than 
the $x$-coordinates of all points in $S_2$. 
After an operation $\insp_x$ or $\delp_x$, we re-build the data 
structure for $H_2$ in $O(1)$ time. 
When the number of elements in a set $H_j$, $j<m$, becomes smaller than 
$2^{2j-1}$ or greater than $2^{2j+1}$, 
we $x$-split the set $H_{j+1}\cup H_j$ into two new sets 
$H'_j$ and $H'_{j+1}$, so that $H'_j$ contains $2^{2j}$ points 
and $H'_{j+1}$ contains $|H_j|+|H_{j+1}|-2^{2j}$ points. 
If the number of elements in $H_m$ exceeds $2^{2m+2}+2^{2m}$ we $x$-split 
$H_j$ into sets $H'_m$ and $H'_{m+1}$ that contain $2^{2m}$ and $2^{2m+2}$ 
points respectively. 
Suppose that the number of points in $H_m$ is smaller than $2^{2m-2}$. 
If $|H_{m-1}|+|H_m|\leq 3\cdot 2^{2m-2}$, we set $H'_{m-1}=H_{m-1}\cup H_m$   
and decrement $m$ by $1$. 
If $|H_{m-1}|+|H_m|>  3\cdot 2^{2m-2}$, we $x$-split $H_{m-1}\cup H_m$ into 
$H'_{m-1}$ and $H'_m$ that contain $2^{2m-2}$ and $|H_{m-1}|+|H_m|-2^{2m-2}$
points respectively. 
We also take care that conditions 2 and 3 are maintained.
If, as a result of $x$-moves, the number of points in some 
$H_j$ that $x$-overlap  $H_{j-1}$ exceeds $2^{2j-3}$, or there is at least 
one point in $H_j$ that $x$-overlaps more than $2^{2j-4}$ points in $H_{j-1}$, 
then we $x$-split 
$H_{j-1}\cup H_{j}$ into sets $H'_{j-1}$ and  $H'_{j}$ that contain $|H_{j-1}|$ 
 and $|H_{j}|$ elements respectively. Each set $H_j$ is rebuilt at most 
once after a sequence of $\Theta(H_j)$ special insert or delete operations. 
Each $H_j$ is also re-built at most once 
 after a sequence of $\Theta(|H_j|)$ $x$-moves, i.e., after $\Theta(H_j)$ 
kinetic events. 
We can maintain 
the list of points in each $H_j$ sorted by their $y$-coordinates; hence, the 
data structure  $D_j$ for a newly re-built set $H_j$ can be constructed in 
$O(|H_j|)$ time. 
Therefore, the amortized  cost of updates and  kinetic events is 
 $O(\log n)$. We can de-amortize update costs using the same techniques as 
in the logarithmic method~\cite{O87}. 
\begin{lemma}
\label{lemma:insx}
There exists a $O(n)$ space data structure that supports 
 dominance 
queries on $U\times U$ grid in $O(\sqrt{\log U/\log \log U} + k)$ time.
Updates after kinetic events and operations $\insp_x$, $\delp_x$, $\insm_x$ 
and $\delm_x$ are supported in $O(\log n)$ time. 
\end{lemma}
{\bf Three-Sided Reporting Queries.}
Now we are ready to describe  data structures that support 
three-sided reporting queries, i.e., the query range is a product of 
a closed interval and a half-open interval.  
We apply the standard method used in e.g.~\cite{VV96},~\cite{SR95}
to augment data structure for dominance queries.

Let $T$ be %the WBB tree~\cite{AV96}  
an arbitrary balanced tree with constant node degree on the set of 
$x$-coordinates of all 
points in $S$. We associate an interval $\halfleftsect{a_l}{b_l}$ with 
each leaf $l$ of $T$, where $a_l$ is the predecessor of  the smallest
 value $m_l$ stored in the
 node $l$, and $b_l$ is the largest value stored in the node $l$. 
We associate an interval $(a_l,+\infty)$ with the rightmost leaf $l$.
With each internal node $v$ of $T$ we associate an interval 
$\inter(v)=\cup \inter(v_i)$ for all children $v_i$ of $v$. 
Let $S_v$ be the set of points $p\in S$ such that $p.x\in \inter(v)$.
In every internal node $v$ we store two data structures $\cL_v$ 
and $\cR_v$ %, implemented as in Lemma~\ref{lemma:insx},
 that support dominance 
reporting queries open to the left and open to the right (i.e., 
queries $\halfleftsect{-\infty}{a}\times\halfleftsect{-\infty}{b}$ and 
$\halfrightsect{a}{+\infty}\times\halfleftsect{-\infty}{b}$) for the 
set $S_v$.  
In each node $v$ we also store the list of points in $S_v$  sorted by their
 $y$-coordinate. 

Given a three-sided query with $Q=[a,b]\times\halfleftsect{-\infty}{c}$ 
we can find in time $O(\sqrt{\log U/\log \log U})$ the node $v$ such that 
$[a,b]\subset \inter(v)$, but $[a,b]\not\subset\inter(v_i)$ for all children
 $v_i$ of $v$. Suppose $\inter(v_j)\subset [a,b]$ for $j=r,r,+1,\ldots,q$. 
Then $x$-coordinates of all points in children $v_r,\ldots, v_q$ of $v$ 
belong to $[a,b]$. We can report all points $p$ in $v_r,\ldots, v_q$ 
whose $y$-coordinate do not exceed $c$ using sorted lists of points 
in $S_{v_r},\ldots, S_{v_q}$.
We also answer two dominance queries $Q_1=\halfrightsect{a}{+\infty}\times
\halfleftsect{-\infty}{c}$ and $Q_2=\halfleftsect{-\infty}{b}\times
\halfleftsect{-\infty}{c}$ with help of data structures $\cR_{v_{r-1}}$ 
and $\cL_{v_{q+1}}$ respectively.

After a kinetic event affecting two points $p$ and $q$ from the same set 
$S_v$, the data structures $\cL_v$ and $\cR_v$ are updated. 
After a kinetic event that affects points $p$ and $q$ that belong 
to two neighbor sets $S_{v_i}$ and $S_{v_{i+1}}$ respectively,
we swap $p$ and $q$:
$p$  is removed from $S_{v_i}$ and inserted 
into $S_{v_{i+1}}$, and  $q$ is removed from $S_{v_{i+1}}$ and inserted 
into $S_{v_i}$. 
In this case a constant number of operations  $\insm_x$ and $\delm_x$ 
(resp.\ $\insp_x$ and $\delp_x$) is performed.
Each point belongs to $O(\log n)$ sets $S_v$. Hence, the space usage 
is $O(n\log n)$ and  an update after a kinetic event takes $O(\log^2 n)$ time. 
\begin{lemma}\label{lemma:3sid}
There exists a $O(n\log n)$ space data structure that supports 
three-sided reporting 
queries for a set of linearly moving points on a $U\times U$ grid in 
$O(\sqrt{\log U/\log \log U}+k)$ time 
and updates after kinetic events in $O(\log^2 n)$ time. 
If trajectories of the points do not change, then the total number of kinetic
 events is $O(n^2)$.
\end{lemma}
{\bf Orthogonal Range Reporting Queries.}
In a similar way to Lemma~\ref{lemma:insx} we can extend the data structure 
to support update operations in two other special cases.
 Let $y_{\min}$ and $y_{\max}$  be the smallest and the largest  $y$-coordinates
 of  points in $S$.
 The operation  $\insp_y$ inserts a point $p$ with $y_{\max} <  p.y$. 
 The operation $\delp_y$ deletes a point $p$ with $p.y=y_{\max}$.
 Operations $\insm_y$ and $\delm_y$ insert and delete a point whose 
 $y$-coordinate is smaller than $y$-coordinates of all other points in $S$. 
Again, it is easy to modify the data structure of Lemma~\ref{lemma:3sid} 
so that it supports $\insm_y$ and $\delm_y$ because these operations 
affect at most one segment of $\cM$. We can support $\insp_y$ and $\delp_y$ 
using the same construction as in Lemma~\ref{lemma:insx}.
 
A point $p$ $y$-overlaps with a point $q$ if $p.y>q.y$.
Analogously to Lemma~\ref{lemma:insx}, points are stored 
in sets $V_2,\ldots, V_m$ and we maintain the invariants: 
 \begin{enumerate}
 \item
 For $i=2,\ldots,m-1$, $V_i$ contains between $2^{2i-1}$ and $2^{2i+2}$ points;
 $V_m$ contains between $2^{2i-2}$ and $2^{2i+2}+2^{2i}$ points
 \item
 Each point of  $V_{i}$  $y$-overlaps at most $2^{2i-4}$ points from $V_{i-1}$ 
 \item
At most $2^{2i-3}$ points from  $V_i$ $y$-overlap with $V_{i-1}$
 \end{enumerate}

Elements of $V_i$ are stored in an augmented  data structure $E_i$ 
of Lemma~\ref{lemma:insx} 
so that kinetic events and operations  $\insm_x$, $\delm_x$, $\insp_x$, 
$\delp_x$, $\insm_y$ and $\delm_y$ are supported. 
For $j=1,\ldots,m$, we maintain $\max_x(j)=\max\{p.x|p\in V_{j}\}$
and $\min_y(j)=\min\{p.y|p\in V_{j}\}$. 
All $\min_y(j)$ are stored in a kinetic binary tree. 
For each $j$,  we store a point $p_j$ with  $p_j.x=\max_x(j)$ and 
$p_j.y=j$ in a kinetic  data structure $X$; 
$X$  supports dominance reporting queries 
in $O(\log \log n + k)$ time and updates in $O(\log n)$ time. 
For each $j$ we also store all points of  $V_j$ in a list $L'_j$ that contains 
all points from $V_j$ in the descending order of their $x$-coordinates.

Given a query $Q=\halfrightsect{a}{+\infty}\times\halfrightsect{b}{+\infty}$,
we can find in $O(\log \log n)$ time the smallest index $j$, such that 
at least one point in $V_j$ has $y$-coordinate smaller than $b$. 
According to conditions 2 and 3 above, $V_j$ and $V_{j+1}$ may contain points 
whose $y$-coordinate are greater than or equal to  $b$. 
The $y$-coordinates of all points in
 $V_2\cup\ldots\cup V_{j-1}$ are greater than $b$. 
The $y$-coordinates of all points in sets $V_i$, $i>j+1$, are smaller than $b$.

Using $X$, we can identify all $V_f$ such that $f< j$ and $V_f$ contains 
at least 
one point $p$ with $p.x \geq a$. For every such $f$, all points $p$ 
such that $p\in H_f$ and $p.x \geq a$ can be reported by traversing the list 
$L'_f$. Hence, reporting all points  $p\in V_f$ such that $f< j$, 
$p.x \geq a$ and $p.y\geq b$ takes $O(\log \log n + k)$ time. 
We can report all points in $H_j$ and $H_{j+1}$ that belong to 
$Q= \halfrightsect{a}{+\infty}\times\halfrightsect{b}{+\infty}$ 
in $O(\log \log U +k)$ time using data structures $E_j$ and $E_{j+1}$ 
respectively. 
 
When a point $p$ is inserted with an operation $\insp_x$ or $\insm_x$,
we identify the data structure $E_i$, such that $p.y > \min_y(i)$ but
$p.y< \min_y(j)$ for all $j<i$, and insert $p$ into $E_i$ as described
in Lemma~\ref{lemma:insx}. When a point $p$ is deleted with operations
$\delp_x$ or $\delm_x$, we delete $p$ from the data structure $E_i$
such that $p\in H_i$. When a point is inserted or deleted with
operations $\insm_y$ or $\delm_y$, we insert or delete this point into
the set $V_m$. If a point is inserted or deleted with operations
$\insp_y$ or $\delp_y$, we rebuild the data structure $E_2$.
Invariants 1-3 for sets $V_i$ can be maintained under $\insp_y$,
$\delp_y$, and kinetic events with the same method that was used in
Lemma~\ref{lemma:insx} to maintain sets $H_i$.
 
Hence, we obtain
\begin{lemma}
\label{lemma:insy}
There exists a data structure that supports 
dominance 
queries on $U\times U$ grid in $O(\sqrt{\log U/\log \log U}+k)$ time.
Kinetic events and operations 
$\insp_x$, $\delp_x$, $\insm_x$, $\delm_x$, $\insp_y$, $\delp_y$, $\insm_y$, 
and $\delm_y$ are supported in $O(\log n)$  time. 
\end{lemma}
We can obtain a $O(n\log n)$ space data structure for three-sided queries 
that supports operations  $\insp_y$, $\delp_y$, $\insm_y$, and $\delm_y$ 
as well as kinetic events in $O(\log^2 n)$ time 
in the same way as in the proof of Lemma~\ref{lemma:3sid}. 
Using the same technique once again, 
we obtain the result for general two-dimensional 
range reporting queries stated in Theorem~\ref{theor:report}.
\begin{theorem}\label{theor:report}
There exists a $O(n\log^2 n)$ space data structure that supports 
orthogonal range reporting queries on $U\times U$ grid 
in $O(\sqrt{\log U/\log \log U}+k)$ time and updates 
after kinetic events in $O(\log^3 n)$  time. 
If trajectories of the points do not change, then the total number of kinetic 
events is $O(n^2)$. 
\end{theorem}
\section{Conclusion and Open Problems}
In this paper we describe a data structure for a set of moving points 
that answers orthogonal range reporting queries in 
$O(\sqrt{\log U/\log \log U}+k)$ time. The query time is dominated 
by the time needed to answer a point location query; the rest of the 
query procedure takes $O(\log \log n + k)$ time. Thus a better algorithm 
for the point location problem would  lead to a better query 
time of our data structure. Proving any $\Omega(\log \log U)$ 
 lower bound for our problem, i.e., proving that kinetic range reporting 
is slower than static range reporting, would be very interesting. 

While kinetic data structures usually support arbitrary changes 
of trajectories and only require that points follow constant-degree 
algebraic  trajectories, we assume that points move linearly and all 
changes of trajectories are known in advance. 
An interesting open question is whether we can construct a kinetic data 
structure that achieves $o(\log n/\log \log n)$ query time without these 
additional assumptions.

\end{document}